\documentclass[11pt,a4paper]{article}

\newif\iffull
\fulltrue

\title{Parameterized $k$-Clustering: The distance matters!}

\iffull
\author{
Fedor V. Fomin\thanks{
Department of Informatics, University of Bergen, Norway.} \addtocounter{footnote}{-1}
\and
Petr A. Golovach\footnotemark{} \addtocounter{footnote}{-1}
\and 
Kirill Simonov\footnotemark{}
}
\else

\usepackage{microtype}%if unwanted, comment out or use option "draft"
\author{Fedor V. Fomin}{Department of Informatics, University of Bergen, Norway}{Fedor.Fomin@uib.no}{}{}
\author{Petr A. Golovach}{Department of Informatics, University of Bergen, Norway}{Petr.Golovach@uib.no}{}{}
\author{Kirill Simonov}{Department of Informatics, University of Bergen, Norway}{Kirill.Simonov@uib.no}{}{}

\authorrunning{F.\,V. Fomin, P.\,A. Golovach, K. Simonov}
\Copyright{Fedor V. Fomin, Petr A. Golovach, Kirill Simonov}

\ccsdesc[100]{Theory of computation~Parameterized complexity and exact algorithms}
\keywords{clustering, parameterized complexity, k-means, k-median}

\EventEditors{}
\EventNoEds{0}
\EventLongTitle{}
\EventShortTitle{}
\EventAcronym{}
\EventYear{}
\EventDate{}
\EventLocation{}
\EventLogo{}
\SeriesVolume{}
\ArticleNo{}
\fi

\usepackage{amsmath,amsfonts,amssymb,amsthm}
\usepackage{mathbbol}
\usepackage{amsthm}
\usepackage{graphicx,color}
\usepackage{boxedminipage}
\usepackage[linesnumbered, boxed, algosection,noline,noend]{algorithm2e}
\usepackage{framed}
\usepackage{thmtools}
\usepackage{thm-restate}
\usepackage{xspace}
\usepackage{todonotes}
\usepackage{pgfplots}
\usetikzlibrary{calc}
\usetikzlibrary{shapes}
\usetikzlibrary{arrows.meta}
\iffull
\usepackage{vmargin}
\setmarginsrb{1in}{1in}{1in}{1in}{0pt}{0pt}{0pt}{6mm}
 \usepackage[pdftex, plainpages = false, pdfpagelabels, 
                 bookmarks=false,
                 bookmarksopen = true,
                 bookmarksnumbered = true,
                 breaklinks = true,
                 linktocpage,
                 pagebackref,
                 colorlinks = true,  
                 linkcolor = blue,
                 urlcolor  = blue,
                 citecolor = red,
                 anchorcolor = green,
                 hyperindex = true,
                 hyperfigures
                 ]{hyperref} 
\else
\usepackage{edtable}
\fi
\usepackage{xifthen}
\usepackage{tabularx}
\usepackage{capt-of}
\usepackage{thmtools}
\usepackage{thm-restate}
\usepackage{subcaption}
\usepackage{booktabs}
\usepackage{calc}

\newcommand{\specialcell}[2][c]{%
  \begin{tabular}[#1]{@{}c@{}}#2\end{tabular}}

\DeclareMathOperator{\operatorClassNP}{{\sf NP}}
\newcommand{\classNP}{\ensuremath{\operatorClassNP}}

\DeclareMathOperator{\operatorClassFPT}{{\sf FPT}\xspace}
\newcommand{\classFPT}{\ensuremath{\operatorClassFPT}\xspace}
\DeclareMathOperator{\operatorClassW}{{\sf W}}
\newcommand{\classW}[1]{\ensuremath{\operatorClassW[#1]}}

\newcommand{\dist}{\operatorname{dist}}
\newcommand{\distp}{\operatorname{dist}_p}
\newcommand{\distone}{\operatorname{dist}_1}
\newcommand{\disttwo}{\operatorname{dist}_2}
\newcommand{\distzero}{\operatorname{dist}_0}
\newcommand{\distinfty}{\operatorname{dist}_\infty}

\newcommand{\Oh}{\mathcal{O}}

\iffull
\theoremstyle{plain}
\declaretheorem[name=Theorem]{theorem}
\newtheorem{lemma}[theorem]{Lemma}

\theoremstyle{definition}
\newtheorem{definition}[theorem]{Definition}

\theoremstyle{remark}

\newtheorem{claim}{Claim}[section]
\newtheorem{proposition}{Proposition}
\fi
\newtheorem{observation}{Observation}

\newcommand{\yes}{{yes}}
\newcommand{\no}{{no}}
\newcommand{\yesinstance}{\yes-instance\xspace}
\newcommand{\noinstance}{\no-instance\xspace}

\newcommand{\pname}{\textsc}
\newcommand{\ProblemFormat}[1]{\pname{#1}}
\newcommand{\ProblemIndex}[1]{\index{problem!\ProblemFormat{#1}}}
\newcommand{\ProblemName}[1]{\ProblemFormat{#1}\ProblemIndex{#1}{}\xspace}

\newcommand{\probClust}{\ProblemName{$k$-Clustering}}
\newcommand{\probClustSelect}{\ProblemName{Cluster Selection}}

\newcommand{\probClique}{\ProblemName{Clique}}
\newcommand{\probMultiClique}{\ProblemName{Multicolored Clique}}

\newcommand{\probOCT}{\ProblemName{Odd Cycle Transversal}}
\newcommand{\probHIOCT}{\ProblemName{Half-Integral Odd Cycle Transversal}}
\newcommand{\probSAT}{\ProblemName{3-SAT}}
\newcommand{\probColoring}{\ProblemName{$k$-Coloring}}



\newlength{\RoundedBoxWidth}
\newsavebox{\GrayRoundedBox}
\newenvironment{GrayBox}[1]%
   {\setlength{\RoundedBoxWidth}{.93\textwidth}
    \def\boxheading{#1}
    \begin{lrbox}{\GrayRoundedBox}
       \begin{minipage}{\RoundedBoxWidth}}%
   {   \end{minipage}
    \end{lrbox}
    \begin{center}
    \begin{tikzpicture}%
       \node(Text)[draw=black!20,fill=white,rounded corners,%
             inner sep=2ex,text width=\RoundedBoxWidth]%
             {\usebox{\GrayRoundedBox}};
        \coordinate(x) at (current bounding box.north west);
        \node [draw=white,rectangle,inner sep=3pt,anchor=north west,fill=white] 
        at ($(x)+(6pt,.75em)$) {\boxheading};
    \end{tikzpicture}
    \end{center}}     

\newenvironment{defproblemx}[2][]{\noindent\ignorespaces%
                                \FrameSep=6pt%
                                \parindent=0pt%
                \vspace*{-1.5em}
                \ifthenelse{\isempty{#1}}{%
                  \begin{GrayBox}{#2}%                
                }{%
                  \begin{GrayBox}{#2 parameterized by~{#1}}%  
                }
                \newcommand\Input{Input:}%                        
                \begin{tabular*}{\textwidth}{@{\hspace{.1em}} >{\itshape} p{1.8cm} p{0.8\textwidth} @{}}%        
            }{
                \end{tabular*}%
                \end{GrayBox}%
                \ignorespacesafterend
            }

\newcommand{\defproblema}[3]{% FJR Version
  \begin{defproblemx}{#1}
    Input:  & #2 \\
    Task: & #3
  \end{defproblemx}
}%

\let\oldnl\nl% Store \nl in \oldnl
\newcommand{\nonl}{\renewcommand{\nl}{\let\nl\oldnl}}% Remove line number for one line

\begin{document}

\date{\today}

\maketitle

%!TEX root = Integer_clustering.tex

\begin{abstract}
We consider the \probClust problem, which is for 
 a given   multiset of $n$ vectors  $X\subset \mathbb{Z}^d$ and a nonnegative number $D$,  to decide whether  $X$ can be partitioned  into $k$ clusters 
 $C_1, \dots, C_k$ such that  the cost 
 \[
 \sum_{i=1}^k \min_{c_i\in \mathbb{R}^d}\sum_{x \in C_i} \|x-c_i\|_p^p \leq D, 
 \] 
where $\|\cdot\|_p$ is the Minkowski ($L_p$) norm of order $p$. 
For $p=1$, \probClust  is the well-known    \textsc{$k$-Median}. 
   For $p=2$, the case of the Euclidean   distance, \probClust  is   \textsc{$k$-Means}.
We show that the  parameterized complexity of \probClust strongly depends on the distance order $p$. In particular, we prove that for every $p\in (0,1]$, \probClust is solvable in time
 $2^{\Oh(D \log{D})} (nd)^{\Oh(1)}$, and hence is fixed-parameter tractable when parameterized by $D$.
 On the other hand, we prove that for distances of orders  $p=0$ and $p=\infty$, no such algorithm exists, unless $\classFPT=\classW1$.

\end{abstract}
\section{Introduction}
 Recall that for $p>0$, the \emph{Minkowski} or \emph{$L_p$-norm} of a vector $x=(x[1], \dots, x[d])\in\mathbb{R}^d$ is defined as 
\[\|x\|_p=\big(\sum_{i=1}^d |x[i]|^p\big)^{1/p}.
\]
Respectively, we define the \emph{($L_p$-norm) distance} between two vectors $x=(x[1], \dots, x[d])$ and $y=(y[1], \dots, y[d])$ as
\[\distp(x,y)=\|x-y\|_p^p=\sum_{i=1}^d|x[i]-y[i]|^p.\] 
 We also consider $\distp$ for $p=0$ and $p=\infty$. For $p=0$, $\distp$ is $L_0$ (or the Hamming) distance, that is the number of different coordinates in $x$ and $y$:  
\[
\distzero(x,y)= |\{i \in \{1,\dots, d\} ~|~ x[i] \ne y[i]\}|.
\]
For $p=\infty$, $\distp$ is $L_\infty$-distance, which is defined as 
\[
\distinfty(x,y)=\max_{i\in \{1, \dots, d\}} |x[i] - y[i]|.
\]

The \probClust problem is defined as follows. For a given   (multi) dataset of $n$ vectors (points)  $X\subset \mathbb{Z}^d$, the task is to find a partition of $X$ into $k$ clusters 
 $C_1, \dots, C_k$ minimizing  the cost 
 \[
 \sum_{i=1}^k \min_{c_i\in \mathbb{R}^d}\sum_{x \in C_i} \distp(x,c_i). 
 \]

In particular, for $p=1$, $\distp$ is the $L_1$-distance and the corresponding clustering problem is known as \textsc{$k$-Median}. (Often in the literature, \textsc{$k$-Median} is also used for  clustering  minimizing the sums of the Euclidean distances.)
 For $p=2$, $\distp$ is the $L_2$ (Euclidean) distance, and then the clustering problem becomes  \textsc{$k$-Means}. 

Let us note that optimal clusterings for the same set of vectors can be drastically different for various values of $p$, as shown in Figure \ref{fig:p1clusterings}. The main conceptual contribution of this paper is that the complexity of \probClust also strongly depends on the choice of $p$. 

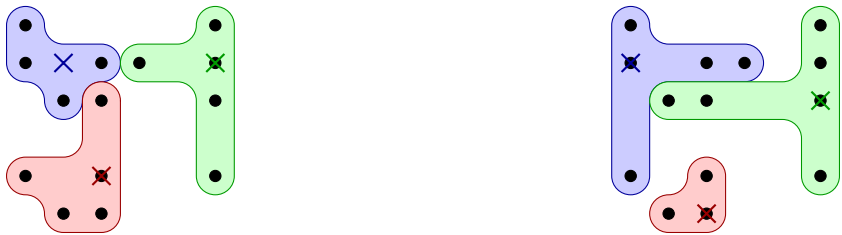
\begin{figure}[ht]
    \centering
    \begin{subfigure}{.5\textwidth}
        \centering
        \begin{tikzpicture}[auto, scale=0.5, node distance=2cm,every loop/.style={},
            base/.style={draw, fill},
            color1/.style={base, circle, inner sep=1.5pt},
            color2/.style={base, rectangle, inner sep=2.1pt, color=blue},
            color3/.style={base, diamond, inner sep=1.35pt, color=green},
            color4/.style={base, regular polygon, regular polygon sides=3, inner sep=0.9pt, color=brown},%, yshift=-1.5pt},
            color5/.style={base, star, star points=10, inner sep=1.1pt, color=magenta},
            cross/.style={base, cross out, inner sep=3pt, thick}
        ]
            \filldraw[fill=blue!20!white, draw=blue!60!black]
            (1, 6.5) arc (90:0:0.5) arc (180:270:0.5)-- (3, 5.5) arc (90:-90:0.5) arc (90:180:0.5) arc(0:-180:0.5) arc(0:90:0.5) arc(-90:-180:0.5) -- (0.5, 6) arc(-180:-270:0.5) --cycle;
            \filldraw[fill=green!20!white, draw=green!60!black]
            (3.5, 5) arc(-180:-90:0.5) --(5, 4.5) arc(90:0:0.5)--(5.5, 2) arc (-180:0:0.5) -- (6.5, 6) arc(0:180:0.5) arc(0:-90:0.5) -- (4, 5.5) arc(90:180:0.5)-- cycle;
            \filldraw[fill=red!20!white, draw=red!60!black]
            (3.5, 4) -- (3.5, 1) arc(0:-90:0.5) -- (2, 0.5) arc(-90:-180:0.5) arc(0:90:0.5) arc(-90:-270:0.5) -- (2, 2.5) arc(-90:0:0.5) -- (2.5, 4) arc(-180:-270:0.5) arc(90:0:0.5) --cycle;
            \node[color1] at (1, 2) (1) {};
            \node[color1] at (6, 4) (2) {};
            \node[color1] at (3, 2) (3) {};
            \node[color1] at (6, 6) (4) {};
            \node[color1] at (4, 5) (5) {};
            \node[color1] at (3, 1) (6) {};
            \node[color1] at (2, 1) (7) {};
            \node[color1] at (1, 5) (8) {};
            \node[color1] at (6, 2) (9) {};
            \node[color1] at (1, 6) (10) {};
            \node[color1] at (3, 4) (11) {};
            \node[color1] at (2, 4) (12) {};
            \node[color1] at (6, 5) (13) {};
            \node[color1] at (3, 5) (14) {};
            \node[cross, color=blue!60!black] at (2, 5) (c1) {};
            \node[cross, color=red!60!black] at (3, 2) (c2) {};
            \node[cross, color=green!60!black] at (6, 5) (c3) {};
        \end{tikzpicture}
    \end{subfigure}\hfill
    \begin{subfigure}{.5\textwidth}
        \centering
        \begin{tikzpicture}[auto, scale=0.5, node distance=2cm,every loop/.style={},
            base/.style={draw, fill},
            color1/.style={base, circle, inner sep=1.5pt},
            color2/.style={base, rectangle, inner sep=2.1pt, color=blue},
            color3/.style={base, diamond, inner sep=1.35pt, color=green},
            color4/.style={base, regular polygon, regular polygon sides=3, inner sep=0.9pt, color=brown},%, yshift=-1.5pt},
            color5/.style={base, star, star points=10, inner sep=1.1pt, color=magenta},
            cross/.style={base, cross out, inner sep=3pt, thick}
        ]
            \filldraw[fill=blue!20!white, draw=blue!60!black]
            (1, 6.5) arc (90:0:0.5) arc (180:270:0.5)-- (4, 5.5) arc (90:-90:0.5) -- (2, 4.5) arc (90:180:0.5) -- (1.5, 2) arc (0:-180:0.5) -- (0.5, 6) arc(-180:-270:0.5) --cycle;
            \filldraw[fill=green!20!white, draw=green!60!black]
            (1.5, 4) arc(-180:-90:0.5) --(5, 3.5) arc(90:0:0.5)--(5.5, 2) arc (-180:0:0.5) -- (6.5, 6) arc(0:180:0.5) -- (5.5, 5)arc(0:-90:0.5) -- (2, 4.5) arc(90:180:0.5)-- cycle;
            \filldraw[fill=red!20!white, draw=red!60!black]
            (3.5, 2) -- (3.5, 1) arc(0:-90:0.5) -- (2, 0.5) arc(-90:-270:0.5) arc(-90:0:0.5) arc(-180:-270:0.5) arc(90:0:0.5) --cycle;
            \node[color1] at (1, 2) (1) {};
            \node[color1] at (6, 4) (2) {};
            \node[color1] at (3, 2) (3) {};
            \node[color1] at (6, 6) (4) {};
            \node[color1] at (4, 5) (5) {};
            \node[color1] at (3, 1) (6) {};
            \node[color1] at (2, 1) (7) {};
            \node[color1] at (1, 5) (8) {};
            \node[color1] at (6, 2) (9) {};
            \node[color1] at (1, 6) (10) {};
            \node[color1] at (3, 4) (11) {};
            \node[color1] at (2, 4) (12) {};
            \node[color1] at (6, 5) (13) {};
            \node[color1] at (3, 5) (14) {};
            \node[cross, color=blue!60!black] at (1, 5) (c1) {};
            \node[cross, color=red!60!black] at (3, 1) (c2) {};
            \node[cross, color=green!60!black] at (6, 4) (c3) {};
        \end{tikzpicture}
    \end{subfigure}
    \caption{Optimal clusterings of the same set of vectors with different distances: $\distone$ in the left subfigure, $\dist_{1/4}$ in the right subfigure. Shapes denote clusters, crosses denote cluster centroids.}
    \label{fig:p1clusterings}
\end{figure}
%We also consider $\distp$ for $p=0$ and $p=\infty$. For $p=0$, $\distp$ is $L_0$ (or Hamming) distance, that is the number of different coordinates in $x$ and $y$  
%\[
%\distzero(x,y)= |\{i \in \{1,\dots, d\} ~|~ x[i] \ne y[i]\}|.
%\]
%For $p=\infty$, $\distp$ is $L_\infty$-distance, which is defined as 
%\[
%\distinfty(x,y)=\max_{i\in \{1, \dots, d\}} |x[i] - y[i]|.
%\]

\probClust, and especially  \textsc{$k$-Median} and \textsc{$k$-Means}, are among the most prevalent problems occurring   in virtually every subarea of data science. We refer to the survey of  Jain \cite{jain2010data} for an extensive overview. 
While in practice the most common approaches to clustering are based on different variations of  Lloyd's heuristic \cite{Lloyd82}, the problem is interesting from the theoretical perspective as well. In particular, there is a vast amount of literature  on approximation algorithms for \probClust   whose behavior can be analyzed rigorously, see e.g.    \cite{AckermannBS10,DBLP:journals/jacm/AgarwalHV04,DBLP:conf/stoc/BadoiuHI02,BoutsidisZMD15,cohen2015dimensionality,DBLP:conf/stoc/FeldmanL11,DBLP:conf/soda/FeldmanSS13,har2004coresets,KumarSS10,delaVegaKKR03,KolliopoulosR07,Cohen-Addad18,SohlerW18}. 

When it comes to  exact solutions,  the complexity  of  %Despite of  the significant progress  in the study of practical and theoretical aspects  
\probClust  is less understood. The  \probClust  problem  is naturally ``multivariate'': in addition to the input size $n$, there are also parameters like space dimension $d$,   number of clusters $k$ or the  cost of clustering $D$.
The problem is known to be \classNP-complete
  for $k = 2$ \cite{AloiseDHP09,Feige14b} and for $d=2$  
 \cite{MegiddoS84,MahajanNV09}. By the  classical work of Inaba et al.~\cite{inaba1994applications},  in the case when both $d$ and $k$ are constants, % are fixed by 
% the classical work of Inaba et al.~\cite{inaba1994applications}
 \probClust is solvable in polynomial  time $\Oh(n^{dk +1})$. 
 Under ETH, the lower bound of $n^{\Omega(k)}$, even when $d = 4$, was shown by Cohen-Addad et al. in~\cite{Cohen-AddadMRR18}  for the
 settings where the set of potential candidate centers is explicitly given as input. However the lower bound of Cohen-Addad et al.  does not generalize to the settings of this paper when any point in Euclidean space can serve as a center.
   For the special case, when the input consists of binary vectors and  the  distance is Hamming, the problem is solvable in time  $2^{\Oh(D \log D)} (nd)^{\Oh(1)}$  \cite{FominGP18}.
%On the other hand, the problem remains \classNP-complete   for $k = 2$ \cite{AloiseDHP09,Feige14b} and for $d=2$  
 %\cite{MegiddoS84,MahajanNV09}. %Thus the problem is solvable in polynomial time 

\medskip\noindent\textbf{Our results and approaches.} In this paper we investigate  the dependence of the complexity of \probClust from the cost of clustering $D$.  It appears, that adding this new ``dimension'' makes the complexity landscape of  \probClust  intricate and interesting. More precisely, we consider the following problem. 

\medskip
\defproblema{\probClust with distance $\dist$}%
{A multiset $X$ of $n$ vectors in $\mathbb{Z}^d$, a positive integer $k$,  and a nonnegative number $D$.}%
{Decide whether there is a partition of $X$ into $k$ clusters $\{C_i\}_{i=1}^k$ and $k$ vectors $\{c_i\}_{i=1}^k$, called  \emph{centroids},  in $\mathbb{R}^d$ such that
    \[\sum_{i=1}^k \sum_{x \in C_i} \dist(x, c_i) \le D.\]
}

 Let us remark that vector set $X$ (like the column set of a matrix) can contain many equal  vectors. Also we consider the situation when vectors from $X$ are integer vectors, while  centroid vectors are not necessarily from $X$.  Moreover, coordinates of centroids can be reals. 

 \medskip
 Our main algorithmic result is the following theorem. 
 \begin{restatable}{theorem}{lmainalgorithmic} \label{thm:lmain_algorithmic}   \probClust with distance $\distp$ 
   is solvable in time
   $2^{\Oh(D \log D)} (nd)^{\Oh(1)}$  for every $p\in (0,1]$.
\end{restatable}
 
 %\todo[inline]{Theorem numbering}

Thus     \probClust  when parameterized by $D$  is fixed-parameter tractable (\classFPT) for Minkowski distance $\distp$ of order $0<p\leq 1$.  Superficially, the general idea of the proof of Theorem~\ref{thm:lmain_algorithmic} is similar to the idea behind the algorithm for \textsc{Binary $r$-Means} for $L_0$ from 
\cite{FominGP18}.  However there are several   differences;  the main   is that the proof in  \cite{FominGP18} is crucially based on the fact that the clustering is performed on binary vectors. Thus the reductions from  \cite{FominGP18} cannot be applied in our case. 
Moreover, as we will see in Theorem~\ref{thm:l0dDhard1}, the existence of an \classFPT algorithm for  \probClust in $L_0$ is highly unlikely. 

In the first step of our algorithm we use color coding to reduce solution of the problem to the    \probClustSelect problem,  which we find interesting on its own. In   \probClustSelect we have $t$ groups of weighted vectors and the task is to select exactly one vector from each group such that the weighted cost of the composite cluster is at most $D$. More formally, 

\medskip

\defproblema{\probClustSelect with distance $\dist$}%
{A set of $m$ vectors $X$ given together with a partition $X=X_1\cup \cdots \cup X_t$ into  $t$ disjoint sets, a weight function $w : X \to \mathbb{Z}_+$, and a nonnegative number $D$.}%
{Decide whether it is possible to select exactly  one vector $x_i$ from each set $X_i$ such that the total cost of the composite cluster formed by  $x_1$, \dots, $x_t$ is at most $D$:
\[\min_{c \in \mathbb{R}^d} \sum_{i=1}^t w(x_i) \cdot \dist(x_i, c) \le D.\]
%    and return the solution $C = \{J_1, \dots, J_t\}$ if it exists.
}

 Informally (see Theorem~\ref{thm:genfpt} for the precise statement), our reduction shows that if the distance norm satisfies some specific properties (which $\distp$ satisfies for all $p$) and if   \probClustSelect is \classFPT parameterized by $D$, then so is \probClust.  
 Therefore,  in order to prove Theorem~\ref{thm:lmain_algorithmic}, all we need is to show that 
\probClustSelect is \classFPT parameterized by $D$ when $p\in (0,1]$. This is the most difficult part of the proof. Here we invoke  the theorem of Marx \cite{Marx08} on the number of subhypergraphs   in hypergraphs of  bounded fractional edge cover.  

Interestingly,   Theorem~\ref{thm:lmain_algorithmic} does not hold for distance $\distzero$. More precisely, for clustering in $L_0$ we prove the following theorem.
\begin{restatable}{theorem}{lzerohard} \label{thm:l0dDhard1}
  %  In the $L_0$ distance,
   With distance $\distzero$,  
     \probClust parameterized by $d+D$ and \probClustSelect parameterized by $d+t+D$  are \classW1-hard.% Moreover,  assuming ETH, there is no $n^{o(d + D^{1/3})}$ algorithm for \probClust and no $n^{o(d + t^{1/2} + D^{1/3})}$ algorithm for \probClustSelect.
\end{restatable}   
In particular, this means that  up to a widely-believed assumption in  complexity that
 $\classFPT\neq \classW1$, Theorem~\ref{thm:l0dDhard1}   rules out    algorithms solving \probClust 
in time $f(d,D)\cdot n^{\Oh(1)}$ and   algorithms solving 
   \probClustSelect   in  $L_0$  in time $g(t,d,D)\cdot n^{\Oh(1)}$ for any functions $f(d,D)$ and $ g(t,d,D)$. 
   Similar hardness result holds   for $L_\infty$. 
   \begin{restatable}{theorem}{linftyhard} \label{thm:linftydDhard1}
  %  In the $L_0$ distance,
  With distance $\distinfty$,  
     \probClust parameterized by $ D$ and \probClustSelect parameterized by $ t+D$  are \classW1-hard.% Assuming ETH, there is no $n^{o(D)}$ algorithm for \probClust and no $n^{o(t + D)}$ algorithm for \probClustSelect.

\end{restatable}    

This naturally brings us to the question:  What happens with  \probClust  for  $p \in (1, \infty)$, especially for the Euclidean distance, that is  $p=2$. Unfortunately, we are not able to answer this question when the parameter is $D$ only. However, we can prove that 
\begin{restatable}{theorem}{ltwoseldD}    \label{thm:l2seldD1} %For $p=2$,  
  \probClust  and    \probClustSelect   with distance $\dist_2$   are  \classFPT when  parameterized by $d+D$.
\end{restatable}
%\todo[inline]{State Theorem \ref{thm:l2seldD1} about \probClust to match the table?}

Thus in particular, Theorem~\ref{thm:l2seldD1} implies that  \probClust     with distance $\dist_2$
 is \classFPT parameterized by $d+D$. On the other hand, we prove that 
 \begin{restatable}{theorem}{lpselecthard}\label{thm:lpselecthard1}
     \probClustSelect with distance $\dist_p$ is \classW1-hard for every $p \in (1, \infty)$ when parameterized by   $t+D$ .% Moreover, unless ETH fails, there is no $n^{o(t^{1/2} + D^{1/2})}$ algorithm for \probClustSelect.

   \end{restatable}
  In particular, Theorem~\ref{thm:lpselecthard1} yields that the approach we used to establish 
  the  tractability (with parameter $D$)  of   \probClust  for  $p =1$ will not work for $p>1$.
   
   We summarize our and previously known algorithmic and hardness results for  \probClust and  \probClustSelect with different distances in Table~\ref{tabl:compl}.
   
  \begin{table}[ht]
%\begin{center}
    \centering
{\small
    \iffull
    \begin{tabular}{|c|c|c|}
    \else
    \begin{edtable}{tabular}{|c|c|c|}
    \fi
\hline
$\distp$   & \probClust  &\probClustSelect\\
\hline
$p=0$ &\specialcell{\classW1-hard param. $d+ D$ [Thm~\ref{thm:l0dDhard1}] \\  \classNP-c for $k = 2$ \cite{Feige14b}}& \classW1-hard param.    $d+ t+ D$  [Thm~\ref{thm:l0dDhard1}] \\
\hline
$0<p\leq 1$ &    
\specialcell{
 $2^{\Oh(D \log D)} (nd)^{\Oh(1)}$ [Thm~\ref{thm:lmain_algorithmic}] \\ \classNP-c for $k = 2$ \cite{Feige14b} \\ \classNP-c for $d = 2$ \cite{MegiddoS84}}
 &  \specialcell{$2^{\Oh(D \log D)} (nd)^{\Oh(1)}$ [Thm~\ref{thm:probClustSelect}] \\ \classW1-hard param. $t+ d$ for $p = 1$ [Thm~\iffull\ref{thm:l1selecthard}\else 12\fi]}\\
\hline
$1<p< +\infty$ &   \specialcell{\classFPT param.  $d+ D$ for $p=2$  [Thm~\ref {thm:l2seldD1}] \\
\classNP-c for $k = 2$ \cite{AloiseDHP09} \\ \classNP-c for $d = 2$ \cite{MahajanNV09}}
 &
 \specialcell{\ \classFPT param.  $d+ D$ for $p=2$  [Thm~\ref {thm:l2seldD1}] \\
 \classW1-hard param. $t+ D$ [Thm~\ref{thm:lpselecthard1}]}\\
\hline
$p=\infty$ &  
\specialcell{\classW1-hard param. $D$ [Thm~\ref{thm:linftydDhard1}]\\  \classNP-c for $k = 2$ [Thm~\iffull\ref{thm:linfoct}\else 15\fi] }&\classW1-hard param. $t+ D$  [Thm~\ref{thm:linftydDhard1}]\\
\hline
\iffull
\end{tabular}
\else
\end{edtable}
\fi
}
\caption{Complexity of \probClust and \probClustSelect. In the table, known \classNP-completeness results are for $p=1$ and $p=2$ only.} \label{tabl:compl}
%\end{center}
\end{table}

%\vspace{-.5cm}
\iffull
 The remaining part of this paper is organized as follows.   Section~\ref{sec:preliminaries} contains preliminaries.
  In  Section~\ref{secfromcluster} we prove Theorem~\ref{thm:genfpt} which provides us with \classFPT Turing reduction from  \probClust to \probClustSelect.  Theorem~\ref{thm:genfpt} appears to be a handy tool to establish tractability of   \probClust. 
 In  Section~\ref{sec:pin01} we collect the results on  clustering with $L_p$-norm for $p\in (0,1]$. 
   In particular, in Subsection~\ref{subsec:FPTD}, we prove Theorem~\ref{thm:lmain_algorithmic}, the main algorithmic result of this work,  stating that when $p \in (0, 1]$, \probClust and \probClustSelect admit FPT algorithms with parameter $D$. In Subsection~\ref{subs:W1clustsel} we complement the algorithmic upper bounds with lower bounds by proving that   \probClustSelect is \classW1-hard when   $p = 1$ and parameter is $t+ d$ (Theorem~\ref{thm:l1selecthard}). 
    In Section~\ref{sec:l0}, we consider the case $p = 0$ and prove 
Theorem~\ref{thm:l0dDhard1} establishing \classW1-hardness of \probClust and  \probClustSelect.
%With distance $\distzero$,  
   %  \probClust parameterized by $d+D$ and \probClustSelect parameterized by $d+t+D$  are \classW1-hard.
Section~\ref{sec:tinfty} is devoted to 
  the case $p = \infty$. Here we  establish two hardness results about \probClust: \classW1-hardness when parameterized by $D$ and \classNP-hardness in the case $k = 2$.
In Section~\ref{sec:pinfty}, we look at the
case $p \in (1, \infty)$, with the particular emphasis on the most commonly used case $p = 2$. We show that when $d+D$ is the parameter, then \probClustSelect and  \probClust in the $L_2$ distance are \classFPT. We also show
that \probClustSelect is \classW1-hard when parameterized by $t + D$ for all $p \in (1, \infty)$.
   We conclude with open problems in Section~\ref{sec:OPEN}. 
\else
In the extended abstract, we provide a full proof of Theorems \ref{thm:lmain_algorithmic} and \ref{thm:probClustSelect}. Proofs of Theorems \ref{thm:l0dDhard1}, \ref{thm:linftydDhard1}, \ref{thm:l2seldD1}, \ref{thm:lpselecthard1}, 12 and 15 can be found in the attached full version of this paper.
\fi

\section{Preliminaries and notation}\label{sec:preliminaries}

%\todo[inline]{add something about the model of computation for $p \notin \{0, 1, 2, \infty\}$?}

\medskip
\noindent\textbf{Cluster notation.}
By a \emph{cluster} we always mean a multiset of vectors in $\mathbb{Z}^d$. For distance $\dist$, the \emph{cost} of a given cluster $C$ is the total distance from all vectors in the cluster to the optimally selected cluster centroid, 
$\min_{c\in \mathbb{R}^d}\sum_{x \in C} \dist(x,c)$.
An \emph{optimal} cluster centroid for a given cluster $C$ is any $c \in \mathbb{R}^d$ minimizing
$\sum_{x \in C} \dist (x,c)$.
For most of the considered distances, we argue that an optimal cluster centroid could always be chosen among selected family of vectors (e.g. integral). Whenever we show this, we only consider optimal cluster centroids of the stated form afterwards.

\medskip\noindent\textbf{Complexity.}
 A \emph{parameterized problem} is a language $Q\subseteq \Sigma^*\times\mathbb{N}$ where $\Sigma^*$ is the set of strings over a finite alphabet $\Sigma$. Respectively, an input  of $Q$ is a pair $(I,k)$ where $I\subseteq \Sigma^*$ and $k\in\mathbb{N}$; $k$ is the \emph{parameter} of  the problem. 
A parameterized problem $Q$ is \emph{fixed-parameter tractable} (\classFPT) if it can be decided whether $(I,k)\in Q$ in time $f(k)\cdot|I|^{\Oh(1)}$ for some function $f$ that depends of the parameter $k$ only. Respectively, the parameterized complexity class \classFPT is composed by  fixed-parameter tractable problems.
The $\operatorClassW$-hierarchy is a collection of computational complexity classes: we omit the technical
definitions here. The following relation is known amongst the classes in the $\operatorClassW$-hierarchy:
$\classFPT=\classW{0}\subseteq \classW{1}\subseteq \classW{2}\subseteq \ldots \subseteq \classW{P}$. It is widely believed that $\classFPT\neq \classW{1}$, and hence if a
problem is hard for the class $\classW{i}$ (for any $i\geq 1$) then it is considered to be fixed-parameter intractable.
We refer to   books \cite{CyganFKLMPPS15,DowneyF13} for the detailed introduction to parameterized complexity.  

\iffull
We also provide conditional lower bounds by making use of  the following complexity hypothesis formulated by  Impagliazzo, Paturi, and Zane   \cite{ImpagliazzoPZ01}.
 
\begin{quote}
\textbf{Exponential Time Hypothesis (ETH)}:  There is a positive real $s$ such that 3-CNF-SAT with $n$ variables and $m$ clauses cannot be solved in time $2^{sn}(n+m)^{\Oh(1)}$.
\end{quote}
 
\medskip
\noindent\textbf{Graphs.}
For  proving \classW1-hardness, we need to consider graphs. Whenever we work with  a graph $G$, we always fix some ordering on the vertices $\pi_V : V(G) \to \{1, \dots, |V(G)|\}$ and on the edges $\pi_E : E(G) \to \{1, \dots, |E(G)|\}$. We drop $\pi_V$ and $\pi_E$ to simplify notation, so when we consider a vertex $v \in V(G)$ or an edge $e \in E(G)$, $v$ and $e$ also denote integers---numbers of $v$ and $e$ according to the orderings $\pi_V$ and $\pi_E$ correspondingly.
\fi

%!TEX root = Integer_clustering.tex
%\section{FPT algorithms parameterized by $D$}
\section{From \probClust   to \probClustSelect}\label{secfromcluster}
In this section we present a general scheme for obtaining an FPT algorithm parameterized by $D$, which is later applied to various distances.

First, we formalize the following intuition: there is no reason to assign equal vectors to different clusters.

\begin{definition}[Initial cluster and regular partition]
    For a multiset of vectors $X$, an inclusion-wise maximal multiset $I \subset X$ such that all vectors in $I$ are equal is called \emph{an initial cluster}.

    We say that a clustering $\{C_1, \dots, C_k\}$ of $X$ is \emph{regular} if for every initial cluster $I$ there is a $i \in \{1, \dots, k\}$ such that $I \subset C_i$.
    \label{def:initreg}
\end{definition}

Now we prove that it suffices to look only for regular solutions.

\begin{proposition}
    Let $(X, k, D)$ be a \yesinstance to \probClust. Then there exists a solution of $(X, k, D)$ which is a regular clustering.
    \label{prop:reg}
\end{proposition}
\begin{proof}
    Let us assume that the instance $(X, k, D)$ has a solution. There are $k$ clusters $\{C_i\}_{i=1}^k$ and $k$ vectors $\{c_i\}_{i=1}^k$ in $\mathbb{R}^d$ such that
\[\sum_{i=1}^k \sum_{x \in C_i} \dist(x, c_i) \le D.\]
    Note that for every $x \in C_j$, $\dist(x, c_j) \ge \min_{1\leq i\leq k} \dist(x, c_i)$. So if we consider a new clustering $\{C_1', \dots, C_k'\}$ with the same centroids, where $C_j'$ are all vectors from $X$ for which $c_j$ is the closest centroid, the total distance does not increase. If we also break ties in favor of the lower index, then for any initial cluster $I$ the same centroid $c_i$ will be the closest, and all vectors from $I$ will end up in $C_i'$, so $\{C_1', \dots, C_k'\}$ is a regular clustering.
\end{proof}

From now on, we consider only regular solutions.

\begin{definition}[Simple and composite clusters]
    We say that a cluster $C$ is \emph{simple} if it is an initial cluster. Otherwise, the cluster is \emph{composite}.
    \label{def:simple}
\end{definition}

Next we state a property of \probClust with a particular distance, which is required for the algorithm. Intuitively, each unique vector adds at least some constant to the cluster cost. In the subsequent sections we show that the property holds for all distances in our consideration.

\begin{definition}[$\alpha$-property]
    We say that a distance has the \emph{$\alpha$-property} for some $\alpha > 0$ if for any $s$ 
    the cost of any composite cluster which consists of $s$ initial clusters is at least $\alpha (s - 1)$.
    \label{def:alpha}
\end{definition}

The following problem is a key subroutine in our algorithm. In some cases it is solvable trivially, but it presents the main challenge for our main algorithmic result in the $L_1$ distance.

\defproblema{\probClustSelect}%
{Family of $t$ disjoint sets of vectors $X_1, \dots, X_t$, containing $m$ vectors in total, a weight function $w : \cup_{i=1}^t X_i \to \mathbb{Z}_+$, and a nonnegative number $D$}%
{Determine whether it is possible to choose one vector $x_i$ from each set $X_i$ such that the total cost of forming a composite cluster out of $x_1$, \dots, $x_t$ is at most $D$:
\[\min_{c \in \mathbb{R}^d} \sum_{i=1}^t w(x_i) \dist(x_i, c) \le D.\]
%    and return the solution $C = \{J_1, \dots, J_t\}$ if it exists.
}

The intuition to the weight function in the definition of \probClustSelect is that it represents sizes of initial clusters, that is, how many equal vectors are there.

We  also need a procedure  to enumerate all possible optimal cluster costs which are less than $D$. It may not be straightforward since not all distances in our consideration are integer. So we assume that the set of all possible optimal cluster costs which are less than $D$ is also given in the input. Now we are ready to state the result formally.

\begin{theorem}
    Assume that the $\alpha$-property holds, \probClustSelect is solvable in time $\Phi(m, d, t, D)$, where $\Phi$ is a non-decreasing function of its arguments, and we are given the set $\mathcal{D}$ of all possible optimal cluster costs which are at most $D$. Then \probClust is solvable in time
\[2^{\Oh(D \log D)} (nd)^{\Oh(1)} |\mathcal{D}| \Phi(n, d, 2D/\alpha, D).\]
    \label{thm:genfpt}
\end{theorem}

\begin{proof}
    By the $\alpha$-property, in any solution there are at most $D/\alpha$ composite clusters, since each contains at least two initial clusters. Moreover, there are at most $2D/\alpha$ initial clusters in all composite clusters.

  Thus by Proposition~\ref{prop:reg},  solving \probClust is equivalent to selecting at most $T := \lceil 2D/\alpha \rceil$ initial clusters and grouping them into composite clusters such that the total cost of these clusters is at most $D$.
We design an algorithm which, taking as a subroutine an algorithm for  \probClustSelect, solves \probClust. 
\iffull
The   algorithm is sketched in Figure~\ref{fig:clustering}, an example is shown in Figure~\ref{fig:clusteringdraw}.
\else
An example is shown in Figure~\ref{fig:clusteringdraw}.
\fi

To perform the selection and grouping, our algorithm uses the color coding technique of Alon, Yuster, and Zwick from \cite{AlonYZ95}. %Our algorithm in the whole is given in the Figure \ref{fig:clustering}.
     Consider the input as a family of initial clusters $\mathcal{I}$. We color initial clusters from $\mathcal{I}$ independently and uniformly at random by $T$ colors 1, 2, \dots, $T$. Consider any solution, and the particular set of at most $T$ initial clusters which are included into composite clusters in this solution.
These initial clusters are colored by distinct colors with probability at least $\frac{T!}{T^{T}} \ge e^{-T}$. Now we construct an algorithm for finding a colorful solution.

\newsavebox\startbox
\begin{lrbox}{\startbox}
    \begin{tikzpicture}[auto, scale=0.5, node distance=2cm,every loop/.style={},
            base/.style={draw, fill},
            color1/.style={base, circle, inner sep=1.5pt, color=red},
            color2/.style={base, rectangle, inner sep=2.1pt, color=blue},
            color3/.style={base, diamond, inner sep=1.35pt, color=green},
            color4/.style={base, regular polygon, regular polygon sides=3, inner sep=0.9pt, color=brown},%, yshift=-1.5pt},
            color5/.style={base, star, star points=10, inner sep=1.1pt, color=magenta},
        ]
            \node[color1] at (0, 0) (1) {};
            \node[color2] at (1, 0) (2) {};
            \node[color5] at (2, 2) (3) {};
            \node[color1] at (0, 3) (4) {};
            \node[color4] at (3, 3) (5) {};
            \node[color4] at (6, 0) (6) {};
            \node[color2] at (7, 2) (7) {};
            \node[color3] at (4, 3) (8) {};
            \node[color1] at (3.5, 0) (9) {};
%            \node[color5] at (-2, 1) (10) {};
            \node[color2] at (-1, 5) (11) {};
            \node[color5] at (5, 5) (12) {};
            \draw[draw opacity=0] (0.5, 0) circle (1cm);
            \draw[draw opacity=0] (0, 3) circle (.5cm);
            \draw[draw opacity=0] (3, 2.5) circle (1.5cm);
            \draw[draw opacity=0] (6, 0) circle (.5cm);
            \draw[draw opacity=0] (7, 2) circle (.5cm);
            \draw[draw opacity=0] (3.5, 0) circle (.5cm);
%            \draw[draw opacity=0] (-2, 1) circle (.5cm);
            \draw[draw opacity=0] (-1, 5) circle (.5cm);
            \draw[draw opacity=0] (5, 5) circle (.5cm);
        \end{tikzpicture}
\end{lrbox}
\newsavebox\ponebox
\begin{lrbox}{\ponebox}
        \begin{tikzpicture}[auto, scale=0.5, node distance=2cm,every loop/.style={},
            base/.style={draw, fill},
            color1/.style={base, circle, inner sep=1.5pt, color=red},
            color2/.style={base, rectangle, inner sep=2.1pt, color=blue},
            color3/.style={diamond, inner sep=1.35pt, color=green},
            color4/.style={regular polygon, regular polygon sides=3, inner sep=0.9pt, color=brown},%, yshift=-1.5pt},
            color5/.style={star, star points=10, inner sep=1.1pt, color=magenta},
        ]
            \node[color1] at (0, 0) (1) {};
            \node[color2] at (1, 0) (2) {};
            \node[color5] at (2, 2) (3) {};
            \node[color1] at (0, 3) (4) {};
            \node[color4] at (3, 3) (5) {};
            \node[color4] at (6, 0) (6) {};
            \node[color2] at (7, 2) (7) {};
            \node[color3] at (4, 3) (8) {};
            \node[color1] at (3.5, 0) (9) {};
 %           \node[color5] at (-2, 1) (10) {};
            \node[color2] at (-1, 5) (11) {};
            \node[color5] at (5, 5) (12) {};
            \draw (0.5, 0) circle (1cm);
            \draw[draw opacity=0] (0, 3) circle (.5cm);
            \draw[draw opacity=0] (3, 2.5) circle (1.5cm);
            \draw[draw opacity=0] (6, 0) circle (.5cm);
            \draw[draw opacity=0] (7, 2) circle (.5cm);
            \draw[draw opacity=0] (3.5, 0) circle (.5cm);
%            \draw[draw opacity=0] (-2, 1) circle (.5cm);
            \draw[draw opacity=0] (-1, 5) circle (.5cm);
            \draw[draw opacity=0] (5, 5) circle (.5cm);
        \end{tikzpicture}
\end{lrbox}
\newsavebox\ptwobox
\begin{lrbox}{\ptwobox}
        \begin{tikzpicture}[auto, scale=0.5, node distance=2cm,every loop/.style={},
            base/.style={draw, fill},
            color1/.style={circle, inner sep=1.5pt, color=red},
            color2/.style={rectangle, inner sep=2.1pt, color=blue},
            color3/.style={base, diamond, inner sep=1.35pt, color=green},
            color4/.style={base, regular polygon, regular polygon sides=3, inner sep=0.9pt, color=brown},%, yshift=-1.5pt},
            color5/.style={base, star, star points=10, inner sep=1.1pt, color=magenta},
        ]
            \node[color1] at (0, 0) (1) {};
            \node[color2] at (1, 0) (2) {};
            \node[color5] at (2, 2) (3) {};
            \node[color1] at (0, 3) (4) {};
            \node[color4] at (3, 3) (5) {};
            \node[color4] at (6, 0) (6) {};
            \node[color2] at (7, 2) (7) {};
            \node[color3] at (4, 3) (8) {};
            \node[color1] at (3.5, 0) (9) {};
%            \node[color5] at (-2, 1) (10) {};
            \node[color2] at (-1, 5) (11) {};
            \node[color5] at (5, 5) (12) {};
            \draw[draw opacity=0] (0.5, 0) circle (1cm);
            \draw[draw opacity=0] (0, 3) circle (.5cm);
            \draw (3, 2.5) circle (1.5cm);
            \draw[draw opacity=0] (6, 0) circle (.5cm);
            \draw[draw opacity=0] (7, 2) circle (.5cm);
            \draw[draw opacity=0] (3.5, 0) circle (.5cm);
%            \draw[draw opacity=0] (-2, 1) circle (.5cm);
            \draw[draw opacity=0] (-1, 5) circle (.5cm);
            \draw[draw opacity=0] (5, 5) circle (.5cm);
        \end{tikzpicture}
\end{lrbox}
\newsavebox\clbox
\begin{lrbox}{\clbox}
        \begin{tikzpicture}[auto, scale=0.5, node distance=2cm,every loop/.style={},
            base/.style={draw, fill},
            color1/.style={base, circle, inner sep=1.5pt, color=red},
            color2/.style={base, rectangle, inner sep=2.1pt, color=blue},
            color3/.style={base, diamond, inner sep=1.35pt, color=green},
            color4/.style={base, regular polygon, regular polygon sides=3, inner sep=0.9pt, color=brown},%, yshift=-1.5pt},
            color5/.style={base, star, star points=10, inner sep=1.1pt, color=magenta},
        ]
            \node[color1] at (0, 0) (1) {};
            \node[color2] at (1, 0) (2) {};
            \node[color5] at (2, 2) (3) {};
            \node[color1] at (0, 3) (4) {};
            \node[color4] at (3, 3) (5) {};
            \node[color4] at (6, 0) (6) {};
            \node[color2] at (7, 2) (7) {};
            \node[color3] at (4, 3) (8) {};
            \node[color1] at (3.5, 0) (9) {};
%            \node[color5] at (-2, 1) (10) {};
            \node[color2] at (-1, 5) (11) {};
            \node[color5] at (5, 5) (12) {};
            \draw (0.5, 0) circle (1cm);
            \draw (0, 3) circle (.5cm);
            \draw (3, 2.5) circle (1.5cm);
            \draw (6, 0) circle (.5cm);
            \draw (7, 2) circle (.5cm);
            \draw (3.5, 0) circle (.5cm);
%            \draw (-2, 1) circle (.5cm);
            \draw (-1, 5) circle (.5cm);
            \draw (5, 5) circle (.5cm);
        \end{tikzpicture}
\end{lrbox}
\newsavebox\circlebox
\begin{lrbox}{\circlebox}
 \begin{tikzpicture}\node[draw,fill,circle, inner sep=1.5pt, color=red] {};\end{tikzpicture}   
\end{lrbox}

\begin{figure}[ht]
    \centering
    \begin{tikzpicture}[
            pic/.style={rectangle, rounded corners=10pt, draw, rectangle split, rectangle split parts=2, text centered, minimum width=5.5cm},
            arr/.style={-{Latex[length=3mm]}, line width=0.5mm},
            scale=0.8
        ]
        \node[pic] (start) {
            A random coloring
            \nodepart{second}
            \usebox\startbox
        };
        %\node[pic, below left=.5cm and -2cm of start] (p1) {
        \node[pic, below right=1cm and -3.5cm of start] (p1) {
            \probClustSelect on \begin{tikzpicture}\node[draw,fill,circle, inner sep=1.5pt, color=red, minimum width=0cm] {};\end{tikzpicture}  and \begin{tikzpicture}\node[rounded corners=0pt, draw,fill,rectangle, inner sep=2.1pt, color=blue, minimum width=0cm] {};\end{tikzpicture}
            \nodepart{second}
            \usebox\ponebox
        };
        \node[pic, below right=-3.5cm and 1.5cm of start] (p2) {
\probClustSelect on \begin{tikzpicture}\node[rounded corners=0pt,draw,fill,diamond, inner sep=1.35pt, color=green, minimum width=0cm] {};\end{tikzpicture}, \begin{tikzpicture}\node[rounded corners=0pt,draw,fill, regular polygon, regular polygon sides=3, inner sep=0.9pt, color=brown, minimum width=0cm] {};\end{tikzpicture} and \begin{tikzpicture}\node[rounded corners=0pt,draw,fill,star, star points=10, inner sep=1.1pt, color=magenta, minimum width=0cm] {};\end{tikzpicture}
            \nodepart{second}
            \usebox\ptwobox
        };
        \node[pic, below right=1cm and -3.5cm of p2] (cl) {
            The resulting clustering
            \nodepart{second}
            \usebox\clbox
        };
        \draw[arr] (start) to (p1);
        \draw[arr] (start) to (p2);
        \draw[arr] (p1) to (cl);
        \draw[arr] (p2) to (cl);

    \end{tikzpicture}
    \caption[]{An illustration to the algorithm in Theorem \ref{thm:genfpt}. We start with a particular random coloring and a particular partition of colors $\mathcal{P} = \{P_1, P_2\}$, where $P_1=\{\text{\begin{tikzpicture}\node[draw,fill,circle, inner sep=1.5pt, color=red] {};\end{tikzpicture}},\text{\begin{tikzpicture}\node[rounded corners=0pt, draw,fill,rectangle, inner sep=2.1pt, color=blue] {};\end{tikzpicture}}\}$ and $P_2 = \{\text{\begin{tikzpicture}\node[rounded corners=0pt,draw,fill,diamond, inner sep=1.35pt, color=green] {};\end{tikzpicture}}, \text{\begin{tikzpicture}\node[rounded corners=0pt,draw,fill, regular polygon, regular polygon sides=3, inner sep=0.9pt, color=brown] {};\end{tikzpicture}}, \text{\begin{tikzpicture}\node[rounded corners=0pt,draw,fill,star, star points=10, inner sep=1.1pt, color=magenta] {};\end{tikzpicture}}\}$. We make two calls to \probClustSelect with respect to $P_1$ and $P_2$ and construct the resulting clustering. In the example, all input vectors are distinct.}
    \label{fig:clusteringdraw}
\end{figure}

We consider all possible ways to split colors between clusters (some colors may be unused). Hence we consider all possible  families $\mathcal{P} = \{P_1, \dots, P_h\}$ of pairwise disjoint non-empty subsets of $\{c \in \{1,\dots, T\} : \exists J \in \mathcal{I} \text{ colored by } c\}$. Each family $\mathcal{P}$  corresponds to a partition  of the set of colors $\{1,\dots, T\}$ if we add one fictitious subset for colors which are not used in the composite clusters. The total number of partitions does not exceed  $T^T = 2^{\Oh(D \log D)}$.

When  partition $\mathcal{P}$ is fixed, we form clusters by solving instances of \probClustSelect: For each $i\in \{1, \dots, h\}$, we take initial clusters colored by elements of $P_i$, bundle together those with the same color, and pass the resulting family to \probClustSelect. First note that there cannot be $P \in \mathcal{P}$ of size at most one, since then \probClustSelect has to make a simple cluster while we assume that all clusters obtained from $\mathcal{P}$ are composite. Second, the total number of clusters has to be $k$, the number of clusters is $|\mathcal{I}| - \sum_{P \in \mathcal{P}} |P| + |\mathcal{P}|$. For each $\mathcal{P}$ we check that both conditions hold, and if not, we discard the choice of $\mathcal{P}$ and move to the next one, before calling the \probClustSelect subroutine.

Next, we formalize how we call the \probClustSelect subroutine. We fix the set of colors $P_i = \{c_1, \dots, c_t\}$, then   take the sets $I_j = \{J \in \mathcal{I}: J\text{ is colored by } c_j\}$ for $j\in\{1,\dots, t\}$. We turn each set of initial clusters $I_j$ into a set of weighted vectors $X_j$ naturally: For each $J \in I_j$, we put one vector $x \in J$ into $X_j$, and $w(x) := |J|$. The family of sets of vectors $X_1$, \dots, $X_t$ and the weight function $w$ are the input for \probClustSelect.
Then we search for the minimum cluster cost bound $d_i \le D$ from $\mathcal{D}$, for which the instance $(X_1, \dots, X_t, d_i)$ of \probClustSelect is a \yesinstance, running each time the algorithm for \probClustSelect.

If for some $i$ setting $d_i$ to $D$ leads to a \noinstance, or if $\sum_{i=1}^h d_i > D$, then we discard the choice of the partition $\mathcal{P}$ and move to the next one.
Otherwise, we report that \probClust has a solution and stop. Next, we prove that in this case the solution indeed exists.
%Otherwise, we reconstruct a solution to \probClust and stop.

%The reconstruction is done as follows: for each $i \in \{1, \dots, h\}$ we make a composite cluster out of the initial clusters returned by the \probClustSelect algorithm. All other clusters are simple, so the total cost is $\sum_{i=1}^h d_i$, which is at most $D$. Thus, if the algorithm finds a solution, then $(X, d, D)$ is a \yesinstance.

We reconstruct the solution to \probClust as follows: For each $i\in\{ 1, \dots, h\}$ the corresponding to $P_i = \{c_1, \dots, c_t\}$ instance of \probClustSelect has a solution $\{x_1, \dots, x_t\}$. For each $j \in \{1, \dots, t\}$, consider the corresponding initial cluster $J_j$ consisting of $w(x_j)$ vectors equal to $x_j$. For each $i \in \{1,\dots,h\}$ we obtain a composite cluster $\cup_{j=1}^t J_j$, all other clusters are simple. So the total cost is $\sum_{i=1}^h d_i$, which is at most $D$. Thus, if the algorithm finds a solution, then $(X, d, D)$ is a \yesinstance.

In the opposite direction.  If there is a solution to \probClust, then there is a regular solution, and with probability at least $e^{-T}$ initial clusters which are parts of composite clusters in this solution are colored by distinct colors. Then, there is a partition $\mathcal{P} = \{P_1, \dots, P_h\}$ which corresponds to this solution. This partition is obtained as follows: put into $P_1$ colors from the first composite cluster, into $P_2$ from the second and so on. At some point our algorithm checks the partition $\mathcal{P}$, and as it finds the optimal cost value for each cluster, then it is at most the cost of the corresponding cluster of the solution from which we started.

To analyze the running time, we consider $2^{\Oh(D\log D)}$ partitions $\mathcal{P}$, for each $\mathcal{P}$ we $|\mathcal{P}| = \Oh(D)$ times search for optimal $d_i$ in time $|\mathcal{D}|$. And for each possible value of $d_i$ we make one call to the \probClustSelect algorithm, which takes time at most $\Phi(n, d, T, D)$.

To amplify the error probability to be at least $1/e$, we do $N = \lceil e^{T}\rceil$ iterations of the algorithm, each time with a new random coloring. As each iteration succeeds with probability at least $e^{-T}$, the probability of not finding a colorful solution after $N$ iterations is at most $(1 - e^{-T})^{e^{T}} \le e^{-1} < 1$. So the total running time is $2^{\Oh(D\log D)}\cdot(nd)^{\Oh(1)} |\mathcal{D}| \Phi(n, d, 2D/\alpha, D)$.

    The algorithm could be derandomized by the standard derandomization technique using perfect hash families \cite{AlonYZ95, NaorSS95}. So \probClust is solvable in the same deterministic time.\end{proof}

\iffull
    \begin{algorithm}[ht]
        \captionof{figure}{\probClust algorithm from Theorem \ref{thm:genfpt}}
        \label{fig:clustering}
        \SetKwInOut{Input}{Input}\SetKwInOut{Output}{Output}
        \nonl\hspace*{-.4cm}\probClust($X$, $k$, $D$, $\alpha$, $\mathcal{D}$)\DontPrintSemicolon\;
        \Input{A multiset $X \subset \mathbb{Z}^d$, a positive integer $k$, real nonnegative values $D$ and $\alpha$, a set $\mathcal{D}$, an algorithm \probClustSelect}
        \Output{\emph{Yes} or \emph{No}}
%        \Output{A partition of $X$ into clusters $C_1$, \dots, $C_k$ with cost at most $D$, or \emph{No solution}}
        \BlankLine
        $T \leftarrow \lceil 2D / \alpha \rceil$\;
        $\mathcal{I} \leftarrow \text{initial clusters of }X$\;
        \For{$\lceil e^T\rceil$ iterations}{
            Fix a random coloring $c$ of $\mathcal{I}$ with colors $\{1, \dots, T\}$\;
            \For{valid partitions $\mathcal{P}$ of $\{1, \dots, T\}$}{
                \For{$i = 1$ \KwTo $|\mathcal{P}|$}{
                    $P_i = \{i_1, \dots, i_t\}$\;
                    \For{$j = 1$ \KwTo $t$}{
                        $X_j \leftarrow \emptyset$\;
                        \For{$J \in \mathcal{I} : c(J) = i_j$}{
                            $x \leftarrow \text{ a point from } J$\;
                            $X_j \leftarrow X_j \cup \{x\}$\;
                            $w(x) \leftarrow |J|$\;
                        }
                    }
                    $d_i \leftarrow D + 1$\;
                    \ForEach{$d \in \mathcal{D}$}{
                        \If{\probClustSelect($X_1$, \dots, $X_t$, $w$, $d$)}{
                            $d_i \leftarrow d$\;
%                            $C_i \leftarrow \text{solution cluster of \probClustSelect}$\;
                            BREAK\;
                        }
                    }
                }
                \If{$\sum_{i = 1}^t d_i \le D$}{
                    %\ForEach{$J \in \mathcal{I} \setminus (C_1 \cup \dots \cup C_{|\mathcal{P}|})$}{
                    %    make new cluster $C_i \leftarrow J$\;
                    %}
                    %$C_1$,\dots, $C_k$ is a solution, STOP\;
                    \emph{Yes}, STOP\;
                }
            }
        }
        \emph{No}, STOP\;
    \end{algorithm}
\else
\fi

%!TEX root = Integer_clustering.tex
\section{Algorithms and complexity for distances with   $p \in (0, 1]$}\label{sec:pin01}

The main motivation for the results in this section is the study of  \probClust with the $L_1$ distance, the case  widely known as \textsc{$k$-Medians}. However, our main algorithmic result also extends to distances of order $p \in (0, 1)$ since in some sense they behave similarly to the $L_1$ distance.

\subsection{FPT algorithm when parameterized by  $D$}\label{subsec:FPTD}

In this subsection, we prove Theorem~\ref{thm:lmain_algorithmic}:   when $p \in (0, 1]$, \probClust admits an FPT algorithm with parameter $D$.   First we state basic geometrical observations for cases $p = 1$ and $p \in (0, 1)$, Then we propose a general algorithm for \probClustSelect which relies only on these properties. Finally, we show how Theorem \ref{thm:genfpt} could be applied.

The next two claims deal with the structure of optimal cluster centroids. We state and prove them in the case of weighted vectors where each vector has a positive integer weight given by a weight function $w$. The unweighted case is just a special case when the weight of each vector is one.
\iffull
\else
The proofs of the claims are straightforward and are available in the full version of this paper.
\fi

First, we show that coordinates of cluster centroids could always be selected among the values present in the input, which helps greatly in enumerating cluster centroids that may be optimal.

\begin{claim} Let  $C= \{x_1, \dots, x_t\}$ be a cluster and   $w : \{x_1, \cdots, x_t\} \to \mathbb{Z}_+$ 
be a weight function. Then there is an optimal (subject to the weighted distance $w(x_i)\cdot \distp(x_i, c)$)  centroid $c$ of $C$ such that  for each   $i \in \{1, \dots, d\}$, the $i$-th coordinate $c[i]$ of the centroid 
is from the values present in the input in this coordinate, that is $c[i]\in  \{x_1[i], \dots, x_t[i]\}$. 
%
% optimal value for a centroid could always be selected among the values present in the input in this coordinate. 
 Moreover, for $p=1$ we may assume that the optimal value is a weighted median of the values present in the $i$-th coordinate.
    \label{claim:presentvalues}
\end{claim}
\iffull
\begin{proof}
   For cluster $C$, consider the corresponding multiset of unweighted vectors $C' = \{x_1, \dots, x_t\}$, where each vector $x \in C$ is repeated $w(x)$ times. We define  $y_j = x_j[i]$ for $j \in \{1, \dots, t\}$. Assume that $y_1 \le y_2 \le \dots \le y_t$.
   Let us  consider an optimal cluster centroid $c$ for $C$ and  denote $z = c[i]$. Figure \ref{fig:p1centroids} shows how the cluster cost behaves with respect to $z$ on a concrete set of values $\{y_i\}$ for $p = 1$ and $p = 1/2$.

\begin{figure}[ht]
    \centering
    \begin{subfigure}{.45\textwidth}
        \centering
        \begin{tikzpicture}[scale=0.8]
            \begin{axis}[xlabel=$z$,ylabel={cost},
                xmin=0,xmax=10,ymin=5,ymax=20, grid=major,no marks,domain=0:10,samples=200]
                                \addplot[very thick] {abs(x - 2) + abs(x - 3) + abs(x - 6) + abs(x - 8)};
            \end{axis}
        \end{tikzpicture}
        \subcaption{$\text{cost}(z) = |z - 2| + |z - 3| + |z - 6| + |z - 8|$}
    \end{subfigure}\hfill
    \begin{subfigure}{.45\textwidth}
        \centering
        \begin{tikzpicture}[scale=0.8]
            \begin{axis}[xlabel=$z$,
                xmin=0,xmax=10,ymin=4,ymax=10, grid=major,no marks,domain=0:10,samples=200]
                                \addplot[very thick] {pow(abs(x - 2), 0.5) + pow(abs(x - 3), 0.5) + pow(abs(x - 6), 0.5) + pow(abs(x - 8), 0.5)};
            \end{axis}
        \end{tikzpicture}
        \subcaption{${\text{cost}(z) = |z - 2|^{1/2} + |z - 3|^{1/2}} + {|z - 6|^{1/2} + |z - 8|^{1/2}}$}
    \end{subfigure}
    \caption{Graphs of cluster cost over different values of $z$: $\distone$ in the left plot, $\dist_{1/2}$ in the right plot. The set of coordinate values is given as $y_1 = 2$, $y_2 = 3$, $y_3 = 6$, $y_4 = 8$.}
    \label{fig:p1centroids}
\end{figure}
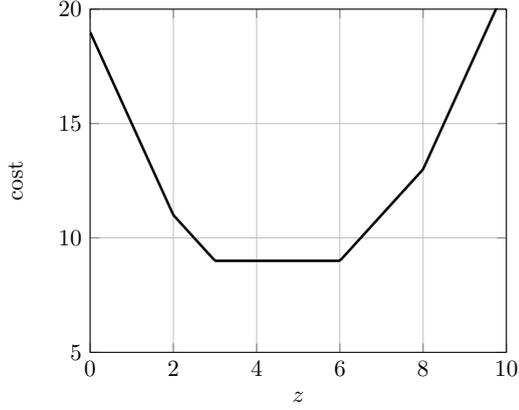
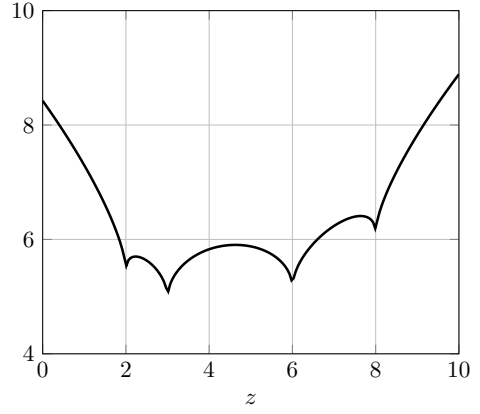
    
For the formal proof, we start with the case $p = 1$. The total cost of $C$ contributed by the $i$-the  coordinate  is
\[|y_1 - z| + |y_2 - z| + \dots + |y_t - z|.\]
 If  $z \in (y_i, y_{i + 1})$ for $i \in \{1, \dots, t - 1\}$, then the derivative with respect to $z$ is
\[((z - y_1) + \dots + (z - y_i) + (y_{i + 1} - z) + \dots + (y_t  - z))' = i - (t - i).\]
    And when $z = y_i$ for $i \in \{1, \dots, t\}$, analogously the derivative is $i - 1 - (t - i)$. So if $t$ is odd, then the derivative is zero at $y_{\lceil t/2\rceil}$, strictly negative before and strictly positive after, so $y_{\lceil t/2 \rceil}$, which is the only median, is the optimal value for $z$. If $t$ is even, then the derivative is zero on $[y_{t/2}, y_{t/2 + 1}]$, strictly negative before and strictly positive after. So any value from $[y_{t/2}, y_{t/2 + 1}]$ is optimal, and we may assume that it is one of the two medians $y_{t/2}$, $y_{t/2 + 1}$.
    
    Now to the case $p \in (0, 1)$, the contribution of the coordinate $i$ is
\[|y_1 - z|^p + |y_2 - z|^p + \dots + |y_t - z|^p.\]
    When  $z$ is between $y_i$ and $y_{i + 1}$, then the derivative of the above with respect to $z$ is equal to
\[p\cdot \left((z - y_1)^{p - 1} + \dots + (z - y_i)^{p - 1} - (y_{i + 1} - z)^{p - 1} - \dots - (y_t - z)^{p - 1}\right).\]
    It is monotone on $(y_i, y_{i + 1})$: when $z$ increases, the sum decreases, as terms of the form $(z - y_j)^{p - 1}$ decrease and terms of the form $(y_j - z)^{p - 1}$ increase, because $p - 1 < 0$. Thus, the optimal value on this interval is achieved at one of its ends. Doing the same for all intervals, we conclude that the optimal value for $z$ must be in $\{y_1, \dots, y_t\}$.
\end{proof}
\else
\fi

In particular, by Claim \ref{claim:presentvalues} we may assume that the coordinates of optimal cluster centroids are integers. Then, the $\alpha$-property holds with $\alpha=1$ since at most one of the initial clusters could have distance zero to the cluster centroid, and all others have distance at least one since the cluster centroid is integral. Namely, let $x$ be a vector in the cluster, and $c$ be the cluster centroid, if $x \ne c$, then there is a coordinate $j$ where $x$ and $c$ differ, and since they are both integral, $|x[j] - c[j]| \ge 1$, and
\[\distp(x,c) = \sum_{i = 1}^d |x[i] - c[i]|^p \ge |x[j] - c[j]|^p \ge 1^p = 1.\]

%Next, we prove that when  at least half of the values by weight equal in some coordinate, the cluster centroid must also have the same value there. The purpose of the claim will be clear when we will discuss the algorithm for \probClustSelect.

In what  follows, the expression \emph{half of vectors by weight}    means that the total weight of the corresponding set of vectors is at least half of the total weight of $C$.

\begin{claim}
    If at least half of the vectors by weight in the cluster $C$ have the same value $z$ in some coordinate $i$, then the optimal cluster centroid is also equal to $z$ in this coordinate.     \label{claim:lphalfcovering}
\end{claim}
\iffull
\begin{proof}
    Let $S$ be the weight-respecting multiset of values which vectors from $C$ have in the $i$-th coordinate: $S = \{x[i] : x \in C, w(x) \text{ times}\}$. Consider the difference between selecting $z$ and some other value $z'$ as the $i$-th coordinate of the centroid:
\[\sum_{y \in S} |y - z|^p - \sum_{y \in S} |y - z'|^p \le \sum_{y \in S, y \ne z} (|y - z|^p - |y - z'|^p - |z - z'|^p).\]
    The inequality holds since at least half of the elements of $S$ are equal to $z$, and so for any value $y \ne z$ there is a term $|z - z'|^p$ in $\sum_{y \in S} |y - z'|^p$ corresponding to one of the values from $S$ equal to $z$. The last sum is non-positive because in every term
\[|y - z|^p \le |y - z'|^p + |z - z'|^p,\]
    as $p \in (0, 1]$.  This concludes the proof. %So $z$ is at least as good as $z'$, for any $z'$.
\end{proof}
\else
\fi

%For more intuition on how the optimal clustering behaves with different distances, examples for $p = 1$ and $p = 1/4$ are shown in Figure \ref{fig:p1clusterings}.

In order to apply Theorem \ref{thm:genfpt},  we need   an FPT algorithm for \probClustSelect. Before obtaining it, we state some properties of hypergraphs, which we need for the algorithm.

A \emph{hypergraph} $G$ is a set of vertices $V(G)$ and a collection of hyperedges $E(G)$, each hyperedge is a subset of $V(G)$. If $G$ and $H$ are hypergraphs, we say that $H$ \emph{appears} at $V' \subset V(G)$ as a \emph{subhypergraph} if there is a bijection $\pi : V(H) \to V'$ with a property that for any $E \in E(H)$ there is $E' \in E(G)$ such that $\pi(E) = E' \cap V'$, the action of $\pi$ is extended to subsets of $V(H)$ in a natural way.

A \emph{fractional edge cover} of a hypergraph $H$ is an assignment $\psi: E(H) \to [0, 1]$ such that for every $v \in V(H)$, $\sum_{E \in E(H) : v \in E} \psi(E) \ge 1$. The \emph{fractional cover number} $\rho^*(H)$ is the minimum of $\sum_{E \in E(H)} \psi(E)$ taken over all fractional edge covers $\psi$.

We need the following result of Marx \cite{Marx08} about finding occurences of one hypergraph in another.

\begin{lemma}[\cite{Marx08}]
    Let $H$ be a hypergraph with fractional cover number $\rho^*(H)$, and let $G$ be a hypergraph where each hyperedge has size at most $\ell$. There is an algorithm that enumerates in time 
    $|V(H)|^{\Oh(|V(H)|)} \cdot \ell^{|V(H)|\rho^*(H)+1} \cdot |E(G)|^{\rho^*(H)+1} \cdot |V(G)|^2$
    every subset $V' \subset V(G)$ where $H$ appears in $G$ as a subhypergraph.
\label{lemma:subhyper}
\end{lemma}

Also, the following version of the Chernoff Bound will be of use.
\begin{proposition}[\cite{AngluinV79}]    \label{prop:chernoff}
    Let $X_1$, $X_2$, \dots, $X_n$ be independent 0-1 random variables. Denote $X = \sum_{i = 1}^n X_i$ and $\mu = E[X]$. Then     for $0 < \beta \le 1$, 
    \begin{gather*}
        P[X \le (1 - \beta)\mu] \le \exp(-\beta^2\mu/2),\\
        P[X \ge (1 + \beta)\mu] \le \exp(-\beta^2\mu/3).
    \end{gather*}
%    for $0 < \beta \le 1$.
\end{proposition}

We are ready to proceed with the proof that   \probClustSelect  with $p\in (0,1]$ is \classFPT when parameterized by $D$. 

\begin{theorem}   \label{thm:probClustSelect}
    For every $p\in (0,1]$, 
    \probClustSelect with distance $\dist_p$ is solvable in time
   $2^{\Oh(D \log D)} (md)^{\Oh(1)}$.
\end{theorem}
\begin{proof}
    First we check if any of the given vectors could be the centroid of the resulting composite cluster. When the centroid is fixed, we find the optimal solution in polynomial time by just selecting the cheapest vector with respect to this centroid from each set. If at some point we find a suitable centroid, then we return that the solution exists. If not, we may assume that the centroid is not equal to any of the given vectors. As a consequence, any vector $x$ selected into the solution cluster contributes at least $w(x)$ to the total distance, since the centroid must be integral by Claim \ref{claim:presentvalues}. So we may now consider only vectors of weight at most $D$ and, moreover, the total weight of the resulting cluster is at most $D$.

    Consider a resulting cluster $C$ with the centroid $c$. There is some $x_1$ in $C$ from $X_1$, and $\distp(x_1, c) \le D$. So if we try all possible $x_1$ from $X_1$ (there are at most $m$ of them), any feasible centroid is at distance at most $D$ from at least one of them. Since $x_1$ and $c$ are integral, they could be different in at most $D$ coordinates, as $\distp(x_1, c) = \sum_{i=1}^d |x_1[i] - c[i]|^p \le D$.

    We try all possible $x_1 \in X_1$. After $x_1$ is fixed, we enumerate all subsets $P$ of coordinates $\{1, \dots, d\}$ where $x_1$ and $c$ could differ, we show how to do it efficiently afterwards. When the subset of coordinates $P$ is fixed, we consider all possible centroids, which are integral, equal to $x_1$ in all coordinates except $P$, and differ from $x_1$ by at most $D^{1/p}$ in each of coordinates from $P$. If $|x_1[i^*] - c[i^*]| > D^{1/p}$ for some coordinate $i^*$, then $\distp(x_1, c) = \sum_{i=1}^d |x_1[i] - c[i]|^p \ge |x_1[i^*] - c[i^*]|^p > D$, so $c$ can not be a centroid. With restrictions stated above, there are at most $2^{\Oh(D\log D)}$ possible centroids. 

It remains to show  that we could enumerate all possible coordinate subsets efficiently. We reduce this task to the task of finding a specific subhypergraph and then apply Lemma \ref{lemma:subhyper}.

    \begin{claim}
        There are $2^{\Oh(D \log D)}$ coordinate subsets where $x_1$ and an optimal cluster centroid $c$ could differ. There exists an algorithm which enumerates all of them in time $2^{\Oh(D \log D)} (md)^{\Oh(1)}$.
        \label{claim:coordinates}
    \end{claim}
    \begin{proof}

        Let $G$ be a hypergraph with $V(G) = \{1,\dots,d\}$, one vertex for each coordinate, and for each vector $x$ in $\cup_{j=1}^t X_j$ we take $w(x)$ multiple hyperedges $E_x$ which contains exactly the coordinates where $x$ and $x_1$ differ. We add an edge only if there are at most $D$ such coordinates, otherwise $x$ can not be in the same cluster as $x_1$. So hyperdeges in $G$ are of size at most $D$. Since we consider only vectors of weight at most $D$, $|E(G)| \le Dm$.

        For a solution, let $x_j$ be the vector selected from the corresponding $X_j$, for $j \in \{1, \dots, t\}$, $C = \{x_1, \dots, x_t\}$ be the solution cluster and $c$ be the centroid. All vectors in $C$ are identical in all coordinates except at most $D$, since if there are different values in at least $D + 1$ coordinates, the cost is at least $D + 1$. Denote this subset of coordinates as $Q$, $c$ could also differ from $x_1$ only at $Q$. Denote the subset of coordinates where $c$ differs from $x_1$ as $P$, $P \subset Q$ and so $|P| \le D$.
        The solution $(C, c)$ induces a subhypergraph $H$ of $G$ in the following way. Leave only hyperedges corresponding to the vectors in $C$, and restrict them to vertices in $P$. There are at most $D$ vertices and at most $D$ hyperedges in $H$, since the total weight is at most $D$. An example of the correspondence between input vectors and hypergraphs is given in Figure~\ref{fig:coordinates}.

        \begin{figure}[ht]
            \centering
            \begin{subfigure}{.5\textwidth}
                \centering
                $D = 2$\par\medskip\begin{tabular}{c | c c c c c}
                    $v$&1&2&\textcolor{red}{3}&4&5\\\hline
                    $\textcolor{red}{x_1}$&0&2&1&3&2\\\hline
                    $x_2$&0&\textbf{1}&1&3&\textbf{1}\\
                    $x_3$&\textbf{1}&2&1&3&\textbf{1}\\
                    $\textcolor{red}{x_4}$&0&2&\textcolor{red}{\textbf{2}}&3&2\\
                    $\textcolor{red}{x_5}$&0&2&\textcolor{red}{\textbf{2}}&3&\textbf{1}\\\midrule
                    $c$&0&2&\textcolor{red}{\textbf{2}}&3&2\\
                \end{tabular}
            \end{subfigure}\hfill
            \begin{subfigure}{.5\textwidth}
                \centering
                \begin{tikzpicture}[auto,  node distance=2cm,every loop/.style={},thick,main node/.style={circle,draw}]
                    \node[main node] (1) {1};
                    \node[main node] (2) [above right of=1] {2};
                    \node[main node,color=red] (3) [below right of=2] {3};
                    \node[main node] (4) [below of =3]{4};
                    \node[main node] (5) [below of=1] {5};
                    \draw[line width=1pt]
                    (2) -- node [pos=0.25] {$x_2$} (5);
                    \draw[line width=1pt]
                    (1) -- node [left] {$x_3$}(5);
                    \draw[line width=1pt, color=red]
                    (3) -- node {$x_5$} (5);
                    \path[line width=1pt, color=red]
                    (3) edge[loop] node {$x_4$} (3);
                \end{tikzpicture}
            \end{subfigure}
            \caption{An illustration of the hypergraph construction in Claim \ref{claim:coordinates}. On the left, the vector $x_1$ and other input vectors $x_2$, \ldots, $x_5$ are given. On the right, the corresponding hypergraph $G$. The solution is in red: on the left, the resulting cluster $\{x_1, x_4, x_5\}$ is of cost 2; on the right, the corresponding subhypergraph is $H$. Note that in $H$ the hyperedge $x_5$ is restricted to the only vertex $3$, so its size is one.}
            \label{fig:coordinates}
        \end{figure}
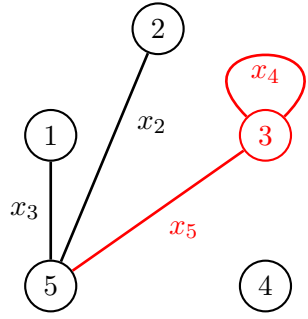
        
        The next claim shows that the fractional cover number of $H$ is bounded by a constant.

    \begin{claim}
        Each vertex in $H$ is covered by at least half of the hyperedges of $H$, and $\rho^*(H) \le 2$.
        \label{claim:fraccover}
    \end{claim}
    \begin{proof}
        Consider a vertex $p \in P$, and assume that less than half of the hyperedges cover $p$.
        It means that in the $p$-th coordinate the centroid $c$ differs from $x_1$, but less than half of the vectors in $C$ by weight differ from $x_1$ in this coordinate. This contradicts Claim \ref{claim:lphalfcovering}.

        So each vertex is covered by at least half of the hyperedges, and setting $\psi \equiv \frac{2}{|E(H)|}$ leads to $\rho^*(H)\le 2$.
    \end{proof}

    In order to enumerate all possible subsets of coordinates $P$, we try all hypergraphs $H$ with at most $D$ vertices and at most $D$ hyperedges, and if each vertex is covered by at least half of the hyperedges, we find all places where $H$ appears in $G$ by Lemma \ref{lemma:subhyper}. The last step is done in $2^{\Oh(D\log D)} \cdot (md)^{\Oh(1)}$ time. However, the number of possible $H$ could be up to $2^{\Omega(D^2)}$. The following claim, which is analogous to Proposition 6.3 in \cite{Marx08}, shows that we could consider only hypergraphs with a logarithmic number of hyperedges.

    \begin{claim}
        If $D \ge 2$, it is possible to delete all except at most $160 \ln D$ hyperedges from $H$ so that in the resulting hypergraph $H^*$ each vertex is covered by at least $1/4$ of the hyperedges, and $\rho^*(H^*) \le 4$.
        \label{claim:logedges}
    \end{claim}
    \begin{proof}
        Denote $s = |E(H)|$, construct a new hypergraph $H^*$ on the same vertex set $V(H)$ by independently selecting each hyperedge of $H$ with probability $(120 \ln D) / s$. Applying Proposition \ref{prop:chernoff} with $\beta = 1/3$, probability of selecting more than $160 \ln D$ hyperedges is at most $\exp((-120 \ln D)/27) < 1 / D^2$.
        By Claim \ref{claim:fraccover}, each vertex $v$ of $H$ is covered by at least $s/2$ hyperedges, and the expected number of hyperedges covering $v$ in $H^*$ is at least $60 \ln D$. By Proposition \ref{prop:chernoff} with $\beta = 1/3$, the probability that $v$ is covered by less than $40 \ln D$
        hyperedges in $H^*$ is at most $\exp(-60 \ln D / 18) \le 1/D^3$. By the union bound, with probability at least $1 - 1/D^2 - D \cdot 1/D^3 > 0$ we select at most $160 \ln D$ hyperedges and each vertex is covered by at least $40\ln D$ hyperedges. So the claim holds, and $\rho^*(H^*) \le 4$ by setting $\psi \equiv \frac{4}{|E(H^*)|}$.
    \end{proof}

    So if there is a subhypergraph $H$ in $G$ corresponding to a solution, then there is also a subhypergraph $H^*$ in $G$ appearing at the same subset of $V(G)$ with at most $160 \ln D$ hyperedges and $\rho^*(H^*) \le 4$. Since we only need to enumerate possible coordinate subsets, it suffices to
    consider only hypergraphs with at most $160 \ln D$ hyperedges, and there are $2^{\Oh(D \log D)}$ of them. Since the fractional cover number is still bounded by a constant, the total running time is $2^{\Oh(D\log D)} \cdot (md)^{\Oh(1)}$, as desired.
    \end{proof}
    With Claim \ref{claim:coordinates} proven, the proof of the theorem is complete.
    \iffull
    The pseudocode  given in Figure~ \ref{alg:clustselect} summarizes the main steps of the algorithm.
    \else
    \fi
    %
    %the description our algorithm, see  is complete now. The pseudocode is given in Figure~ \ref{alg:clustselect}.
    \end{proof}

    \iffull
    \begin{algorithm}[h]
        \captionof{figure}{\probClustSelect algorithm from Theorem \ref{thm:probClustSelect}}
        \label{alg:clustselect}
        \nonl\hspace*{-.4cm}\probClustSelect($X_1$, \dots, $X_t$, $w$, $D$)\DontPrintSemicolon\;
        \SetKwInOut{Input}{Input}\SetKwInOut{Output}{Output}
        \Input{Sets of vectors $X_1$, \dots, $X_t$, a weight function $w$, a nonnegative integer $D$}
        %\Output{A cluster $C$ which contains exactly one initial cluster from each $I_j$ and has a cost at most $D$, or \emph{No solution}}
        \Output{\emph{Yes} or \emph{No}}
        \BlankLine

        \For{vector $c$ in the input}{
            \If{$\sum_{i=1}^t \min_{x_i \in X_i} w(x_i) \distp(x_i, c) \le D$}{
                \emph{Yes}, STOP\;
            }
        }
        \For{$x_1 \in X_1$}{
            $G \leftarrow \text{hypergraph with } V(G) = \{1, \dots, d\}, E(G) = \{\text{positions where $x_1$ and $x$ differ} : x \in \cup_{j = 1}^t X_j, w(x) \text{ times}\}$\;
            \For{hypergraph $H$ with at most $D$ vertices and at most $160 \ln D$ hyperedges}{
                \If{each vertex of $H$ is covered by at least 1/4 of its hyperedges}{
                    \For{place $P$ where $H$ appears in $G$ as subhypergraph}{
                        \For{integer vector $c$ which differs from $x_1$ only at $P$ by at most $D^{1/p}$}{
                            \If{$\sum_{i=1}^t \min_{x_i \in X_i} w(x_i) \distp(x_i, c) \le D$}{
                            \emph{Yes}, STOP\;
                        %        $\{\argmin_{J_i \in I_i} \sum_{x \in J_i} ||x - c||_1\}_{i=1}^t$ is a solution, STOP\;
                            }
                        }
                    }
                }
            }
        }
        \emph{No}, STOP\;
    \end{algorithm}
    \else
    \fi

Combining Theorem \ref{thm:genfpt} and Theorem \ref{thm:probClustSelect}, we obtain an \classFPT algorithm for \probClust. This proves Theorem \ref{thm:lmain_algorithmic}, which we recall here.

\lmainalgorithmic*
\begin{proof}
    We have an algorithm for \probClustSelect whose running time is specified by  Theorem \ref{thm:probClustSelect}.
    By Claim \ref{claim:presentvalues}, the $\alpha$-property holds. 
    The only missing part is to describe the way of producing  the set $\mathcal{D}$ of all possible cluster costs which are at most $D$.

    In the case $p=1$ all distances are integral so we  can  take $\mathcal{D} = \{0, \dots, D\}$.

    For the general case, let  $\mathcal{B} = \{a^p : a \in \{1, \dots, \lceil D^{1/p} \rceil\}\}$. Consider a cluster $C = \{x_1, \dots, x_t\}$ and the corresponding optimal cluster centroid $c$. For any $x_j \in C$, $\distp(x_j, c) = \sum_{i=1}^d |x_j[i] - c[i]|^p$ is a combination of elements of $\mathcal{B}$ with nonnegative integer coefficients.  This is because $x_j$ and $c$ are integral and the cluster cost is at most $D$, hence $|x_j[i] - c[i]| \le D^{1/p}$ for each $i \in \{1, \dots, d\}$. Since weights are also integral, the whole cluster cost is a combination of distances between cluster vectors and the centroid with nonnegative integer coefficients, and so also a combination of elements of $\mathcal{B}$ with nonnegative integer coefficients. This means that we can take
\[\mathcal{D} = \left\{\sum_{b \in \mathcal{B}} a_b \cdot b : a_b \in \mathbb{Z}, a_b \ge 0, \sum_{b \in \mathcal{B}} a_b \le D\right\},\]
    the sum of coefficients $a_b$ is at most $D$ since all elements of $\mathcal{B}$ are at least 1. The size of $\mathcal{D}$ is at most $|\mathcal{B}|^{D} = 2^{\Oh(D \log D)}$.
\end{proof}

Note that another widely studied version of \probClust is where centroids $c_i$ could be selected only among the set of given vectors. Naturally, our algorithm also works in this setting since the set of possible centroids is only restricted further.

\iffull
\subsection{W[1]-hardness of \probClustSelect parameterized by  $t+ d$ for $p = 1$}\label{subs:W1clustsel}

In this subsection, we restrict our attention to the $p = 1$ case. What happens when $D$ is not bounded, but the dimension $d$ and the number of clusters $k$ are parameters? There is a trivial XP-algorithm in time $n^{\Oh(kd)}$, as by Claim \ref{claim:presentvalues} it suffices to try all possible combinations of the values present in coordinates as possible cluster centroids. There are at most $n$ distinct values in each coordinate, so at most $n^d$ candidates for a cluster centroid. After the cluster centroids are fixed, each vector goes to the cluster with the closest centroid.

    We do not know of a lower bound for \probClust complementing this algorithm. However, we are able to show the hardness of \probClustSelect with respect to the dimension.

    \begin{theorem}
        \probClustSelect with distance $\distone$ is \classW1-hard when parameterized by $t + d$.% Assuming ETH, there is no $n^{o(d + t^{1/2})}$ algorithm for \probClustSelect.
        \label{thm:l1selecthard}
    \end{theorem}
    \begin{proof}
        We construct a reduction from \probMultiClique with the input $G$ and $k$. We set $d$ to $k$, for each pair of colors $1 \le i < j \le k$ and each $e = \{u, v\}$ between a vertex $u$ of color $i$ and a vertex $v$ of color $j$ we add a vector $x_e$ to the set $X_{i, j}$, such that $x_e[i] = u$, $x_e[j] = v$
        and all other coordinates are set to zero, and a vector $y_e$ to the set $Y_{i, j}$ which is the same as $x_e$, only coordinates other that $i$ and $j$ are set to $|V(G)| + 1$. We will refer to 0 and $|V(G)| + 1$ as boundary values. The sets $X_{i, j}$ and $Y_{i, j}$ are the input to \probClustSelect, so $t$ is $2\binom{k}{2}$,
        and we set $D$ to $k(|V(G)| + 1)\binom{k - 1}{2}$. Intuitively, the set $X_{i, j}$ corresponds to the choice of the clique edge between $i$-th and $j$-th color, and $Y_{i, j}$ mirrors it. All vectors have weight one. An example is given in Figure \ref{fig:l1selecthard}.

        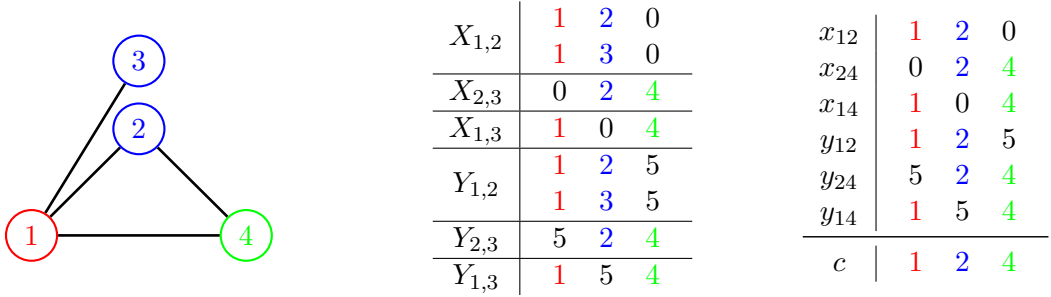
\begin{figure}[ht]
            \centering
            \begin{subfigure}{0.4\textwidth}
                \centering

                \begin{tikzpicture}[auto,  node distance=2cm,every loop/.style={},thick,main node/.style={circle,draw}, baseline={(0,1)}]
                    \node[main node, color=red] at (0,0) (1) {1};
                    \node[main node, color=blue] (2) [above right of=1] {2};
                    \node[main node,color=blue] (3) [above =0.2cm of 2] {3};
                    \node[main node,color=green] (4) [below right of=2] {4};
                    \draw[line width=1pt]
                    (1) -- (2);
                    \draw[line width=1pt]
                    (1) -- (3);
                    \draw[line width=1pt]
                    (1) -- (4);
                    \draw[line width=1pt]
                    (2) -- (4);
                \end{tikzpicture}
            \end{subfigure}\hfill
            \begin{subfigure}{0.3\textwidth}
                \centering
                \begin{tabular}{c|c}
                    $X_{1, 2}$& \begin{tabular}{c c c}
                        \textcolor{red}{1}&\textcolor{blue}{2}&0\\
                        \textcolor{red}{1}&\textcolor{blue}{3}&0
                    \end{tabular}\\
                    \hline
                    $X_{2, 3}$& \begin{tabular}{c c c}
                        0&\textcolor{blue}{2}&\textcolor{green}{4}
                    \end{tabular}\\
                    \hline
                    $X_{1, 3}$& \begin{tabular}{c c c}
                        \textcolor{red}{1}&0&\textcolor{green}{4}
                    \end{tabular}
                    \\
                    \hline
                    $Y_{1, 2}$& \begin{tabular}{c c c}
                        \textcolor{red}{1}&\textcolor{blue}{2}&5\\
                        \textcolor{red}{1}&\textcolor{blue}{3}&5
                    \end{tabular}\\
                    \hline
                    $Y_{2, 3}$& \begin{tabular}{c c c}
                        5&\textcolor{blue}{2}&\textcolor{green}{4}
                    \end{tabular}\\
                    \hline
                    $Y_{1, 3}$& \begin{tabular}{c c c}
                        \textcolor{red}{1}&5&\textcolor{green}{4}
                    \end{tabular}
                \end{tabular}
            \end{subfigure}\hfill
            \begin{subfigure}{0.3\textwidth}
                \centering
                \begin{tabular}{c|c}
                    $x_{12}$& \begin{tabular}{c c c}
                        \textcolor{red}{1}&\textcolor{blue}{2}&0\\
                    \end{tabular}\\
                    $x_{24}$& \begin{tabular}{c c c}
                        0&\textcolor{blue}{2}&\textcolor{green}{4}
                    \end{tabular}\\
                    $x_{14}$& \begin{tabular}{c c c}
                        \textcolor{red}{1}&0&\textcolor{green}{4}
                    \end{tabular}
                    \\
                    $y_{12}$& \begin{tabular}{c c c}
                        \textcolor{red}{1}&\textcolor{blue}{2}&5
                    \end{tabular}\\
                    $y_{24}$& \begin{tabular}{c c c}
                        5&\textcolor{blue}{2}&\textcolor{green}{4}
                    \end{tabular}\\
                    $y_{14}$& \begin{tabular}{c c c}
                        \textcolor{red}{1}&5&\textcolor{green}{4}
                    \end{tabular}\\\midrule
                    $c$&\begin{tabular}{c c c}
                        \textcolor{red}{1}&\textcolor{blue}{2}&\textcolor{green}{4}\\
                    \end{tabular}
                \end{tabular}
            \end{subfigure}

            \caption{An example illustrating the reduction in Theorem \ref{thm:l1selecthard}: an input graph $G$ with vertices colored in three colors, the sets of vectors produced by the reduction, and the resulting optimal cluster, corresponding to the clique on $\{1, 2, 4\}$.}
            \label{fig:l1selecthard}
        \end{figure}

        Note that in any feasible cluster, each coordinate $i$ has exactly $2(k - 1)$ values in $[1, |V(G)|]$, one from each of the sets $X_{i, j}$ and $Y_{i, j}$ for $j \ne i$. Out of all $2(\binom{k}{2} - k + 1) = 2\binom{k - 1}{2}$ other values, exactly half are zero and half are $|V(G)| + 1$. So the median is always in $[1, |V(G)|]$, and the boundary values in each column contribute exactly $(|V(G)| + 1)\binom{k - 1}{2}$ to the total distance.

        Assume there is a colorful $k$-clique in $G$, with vertices $v_1$, $v_2$, \dots, $v_k$.
        We form the resulting cluster by choosing the vector corresponding to the clique's edge between its $i$-th and $j$-th vertices from $X_{i, j}$, and also from $Y_{i, j}$, for all $1 \le i < j \le k$.
        For this cluster, in the $i$-th coordinate we have all non-boundary values equal to $v_i$. So the median is also $v_i$, and the total distance is $D$, since non-boundary values do not contribute anything.

        In the other direction, if we are able to select a cluster of cost exactly $D$, then all non-boundary values in each coordinate must be equal, denote this common value in the $i$-th coordinate as $v_i$. We claim that vertices $v_1$, $v_2$, \dots, $v_k$ form a colorful clique in $G$. Indeed, since we have $2(k - 1)$ times $v_i$ in the $i$-th column, then we have $(k - 1)$ of them from the sets $X_{i, j}$, one from each, and in the $j$-th column the only non-boundary value is $v_j$. So $v_i$ must have an edge to each $v_j$ for $j \ne i$. By construction, vertices in the $i$-th coordinate are of color $i$.

%        Finally, to see that the ETH bound holds note that an $n^{o(d + t^{1/2})}$ algorithm for \probClustSelect together with the reduction above would give a $|V(G)|^{o(k)}$ algorithm for \probMultiClique since $n = |E(G)|$, $d = k$ and $t = 2\binom{k}{2}$.
    \end{proof}

%In this section we consider the $L_p$ distance, which is defined as
%$$||x - y||_p^p = \sum_{i = 1}^d |x_i - y_i|^p,\]
%for any real $0 < p < 1$. Although it is not as nice as $L_1$ distance --- the mean is not necessarily the median, two crucial properties still hold, and we are still able to solve the problem in time $2^{\Oh(D \log D)} \cdot (nm)^{\Oh(1)}$ analoguosly to the $L_1$ case.

\fi

\iffull
%!TEX root = Integer_clustering.tex
\section{The $L_0$ distance}\label{sec:l0}

In this section, we consider the case $p = 0$. It is a natural measure of difference to consider since observation parameters are often incomparable, and we very well may be interested in counting only the number of different entries. From another point of view, the $L_0$ distance gives the \probClust problem a more combinatorial flavor, since the input vectors could be viewed as strings and we are interested about how close they are according to the Hamming distance. However, in comparison to a number of problems on strings, the size of the alphabet is unbounded.

First, note that there is a simple rule of finding the optimal cluster centroid for a given cluster.

\begin{observation}
    For a given cluster $C$, the coordinates of the optimal cluster centroid $c$ could be set as
\[c[i] = \text{the most frequent element of the multiset } \{x[i]\}_{x \in C}, \ 1 \le i \le d,\]
    breaking ties in favor of the lowest values.
    \label{obs:l0frequents}
\end{observation}

By Observation \ref{obs:l0frequents}, we may assume that optimal cluster centroids could never have values not present in the input, and in particular that they are integral.

We prove W[1]-hardness of \probClust with the $L_0$ distance by showing a reduction from \probClique. The reduction also shows hardness of \probClustSelect.

Note that when $d$ is fixed, we could apply Theorem \ref{thm:genfpt} to obtain an FPT algorithm: \probClustSelect solves trivially by trying every present value in each coordinate as a value for the centroid, there are only $n^d$ variants. The $\alpha$-property holds for $L_0$ distance with $\alpha = 1$ since at most one initial cluster could coincide with the cluster centroid, and all others have distance at least one.

We restate Theorem~\ref{thm:l0dDhard1}, which we prove next.

\lzerohard*
\begin{proof}
    First we show how to obtain an FPT reduction from \probClique parameterized by the clique size to \probClust.

    Given an instance ($G$, $k$) of \probClique, for each pair of indices $\{i, j\}$, $1 \le i < j \le k$, we make $|E(G)|$ vectors in $\mathbb{Z}^k$, assume $k \ge 3$. For each $e = \{u, v\} \in E(G)$, we add a vector $x_{i, j, e}$: two coordinates are set to vertex values, $x_{i, j, e}[i] = u$, $x_{i, j, e}[j] = v$, and in all other coordinates $x_{i, j, e}$ is set to the special padding value $c_{i, j, e}  = |V(G)| + (k\cdot i + j) \cdot |E(G)| + e$.
    In total, there are $n = \binom{k}{2}|E(G)|$ vectors and $|V(G)| + \binom{k}{2}|E(G)|$ different values, since there are $|V(G)|$ vertex values, all padding values are distinct from vertex values and from each other.

    Finally, we set $k' = n - \binom{k}{2} + 1$ and $D = \binom{k}{2}(k - 2)$. An example of the reduction is shown in Figure \ref{fig:l0dDhard1}.

        \begin{figure}[ht]
            \centering
            \begin{subfigure}{0.4\textwidth}
                \centering

                \begin{tikzpicture}[auto,  node distance=2cm,every loop/.style={},thick,main node/.style={circle,draw}, baseline={(0,1)}]
                    \node[main node] at (0,0) (1) {1};
                    \node[main node] (2) [above right of=1] {2};
                    \node[main node] (3) [above =0.2cm of 2] {3};
                    \node[main node] (4) [below right of=2] {4};
                    \draw[line width=1pt]
                    (1) -- (2);
                    \draw[line width=1pt]
                    (1) -- (3);
                    \draw[line width=1pt]
                    (1) -- (4);
                    \draw[line width=1pt]
                    (2) -- (4);
                \end{tikzpicture}
            \end{subfigure}\hfill
            \begin{subfigure}{0.3\textwidth}
                \centering
                \begin{tabular}{c|c}
                    $x_{1, 2, \cdot}$& \begin{tabular}{c c c}
                        1&2&$\cdot$\\
                        1&3&$\cdot$\\
                        1&4&$\cdot$\\
                        2&4&$\cdot$\\
                    \end{tabular}\\
                    \hline
                    $x_{1, 3, \cdot}$& \begin{tabular}{c c c}
                        1&$\cdot$&2\\
                        1&$\cdot$&3\\
                        1&$\cdot$&4\\
                        2&$\cdot$&4\\
                    \end{tabular}\\
                    \hline
                    $x_{2, 3, \cdot}$& \begin{tabular}{c c c}
                        $\cdot$&1&2\\
                        $\cdot$&1&3\\
                        $\cdot$&1&4\\
                        $\cdot$&2&4\\
                    \end{tabular}
                \end{tabular}
            \end{subfigure}\hfill
            \begin{subfigure}{0.3\textwidth}
                \centering
                \begin{tabular}{c|c}
                    $x_{1,2,12}$& \begin{tabular}{c c c}
                        1&2&$\cdot$\\
                    \end{tabular}\\
                    $x_{1,3,14}$& \begin{tabular}{c c c}
                        1&$\cdot$&4
                    \end{tabular}\\
                    $x_{2,3,24}$& \begin{tabular}{c c c}
                        $\cdot$&2&4
                    \end{tabular}
                    \\\midrule
                    $c$&\begin{tabular}{c c c}
                        1&2&4\\
                    \end{tabular}
                \end{tabular}
            \end{subfigure}

            \caption{An example illustrating the reduction in Theorem \ref{thm:l0dDhard1}: an input graph $G$, the vectors produced by the reduction (for clarity, they are separated over corresponding pairs $\{i, j\}$, and padding values are replaced by dots), and the only composite cluster in the resulting optimal clustering of cost 3, corresponding to the clique on $\{1, 2, 4\}$.}
            \label{fig:l0dDhard1}
        \end{figure}
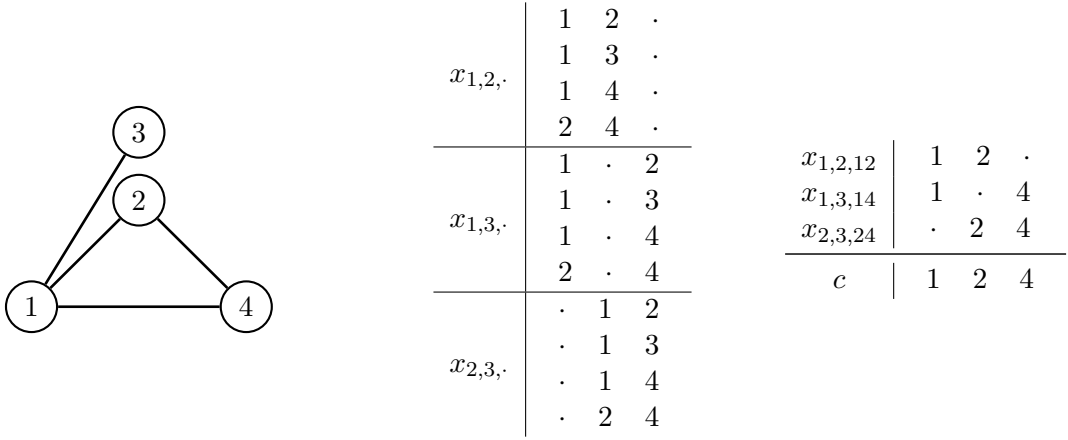

    Now we prove that the original instance has a $k$-clique iff the transformed instance has a $k'$-clustering of cost at most $D$.

If there is a $k$-clique, there is a clustering with cost $D$: we take one nontrivial cluster of size $\binom{k}{2}$ and all other clusters are of size 1. Let $v_1$,..., $v_k$ be the vertices of the clique, for each $\{i, j\}$, $1 \le i < j\le k$ we take $x_{i, j, \{v_i, v_j\}}$ into the cluster. The cluster centroid is $(v_1, ..., v_k)$, each vector in the cluster has distance to the centroid of exactly $(k - 2)$.

Now to the opposite direction. Assume that there is a clustering of cost at most $D$, and there are $t$ composite clusters: $C_1$, ..., $C_t$. In each cluster and each coordinate, by Observation \ref{obs:l0frequents} we may assume that we select the most frequent vertex there as the value of the centroid, since all padding values are distinct. If there are no vertex values in this cluster in this coordinate, we may assume that we select any of the occuring padding values.
For a cluster $C$, denote the number of vertex-containing coordinates as $\beta(C)$, and the total number of vertex-valued entries which do not match with the centroid value in the corresponding coordinate as $\gamma(C)$. We could write the total cost of the clustering as
\[\sum_{i=1}^t \left(|C_i| (k - 2) - (k - \beta(C_i)) + \gamma(C_i) \right).\]
That holds since in each cluster $C_i$ each of the $|C_i| (k - 2)$ padding values is not matched with the cluster centroid and increases the total distance by one, except for the $(k - \beta(C_i))$ vertex-free coordinates, where exacly one of the padding values is selected as the value of the centroid. 
Also each vertex-valued entry which is not matched with the centroid increases the total distance by one, there are $\gamma(C_i)$ of them.

There are $n - \binom{k}{2} + 1$ clusters in total, $n - \binom{k}{2} + 1 - t$ of them are simple. We may assume that in the optimal clustering there are no empty clusters, since we could always move a vector from a composite cluster to an empty one without increasing the cost. So there are $n - (n - \binom{k}{2} + 1 - t) = t + \binom{k}{2} - 1$ vectors in the composite clusters, which is equal to $\sum_{i=1}^t |C_i|$. We could rewrite the total cost as
\[(t + \binom{k}{2} - 1) (k - 2) - tk + \sum_{i=1}^t (\beta(C_i) + \gamma(C_i)) = \binom{k}{2}(k - 2) - (k - 2) + \sum_{i=1}^t (\beta(C_i) - 2 + \gamma(C_i)).\]

Now we show that for any clustering the value $\sum_{i=1}^t (\beta(C_i) - 2 + \gamma(C_i))$ is at least $(k - 2)$, and it is equal to $(k-2)$ only in the $k$-clique clustering.
It suffices to prove the following lemma. 
\begin{lemma}
    For any cluster $C$ such that $2 \le |C| \le \binom{k}{2}$, $\frac{\beta(C) - 2 + \gamma(C)}{|C| - 1} \ge \kappa$, where $\kappa = \frac{k - 2}{\binom{k}{2} - 1} = \frac{2}{k+1}$, and the equality holds only when $C$ is a $k$-clique.
    \label{lemma:l0clustercost}
\end{lemma}

The lemma implies
\[\sum_{i=1}^t (\beta(C_i) - 2 + \gamma(C_i)) = \sum_{i=1}^t \frac{\beta(C_i) - 2 + \gamma(C_i)}{|C_i| - 1} (|C_i| - 1) \ge \kappa \sum_{i = 1}^t (|C_i| - 1) = \kappa \left(\binom{k}{2} - 1\right) = k - 2,\]
and also that the equality holds only when each term is equal to $\kappa$, so each $C_i$ is a $k$-clique, but then $t = 1$ since $\sum_{i= 1}^t (|C_i| - 1) = \binom{k}{2} - 1$. So $G$ must contain a $k$-clique if there is a clustering of cost at most $D$, and the reduction is correct. Note that none of the $C_i$ could have size larger than $\binom{k}{2}$ since there are $n - \binom{k}{2} + 1$ clusters in total.

\begin{proof}[Proof of Lemma \ref{lemma:l0clustercost}]
    
    First, we consider the case $\gamma(C) = 0$, so in each coordinate all vertex values are equal.

    \begin{claim}
        If $C$ is a cluster of vectors obtained by applying the reduction described in the proof of Theorem~\ref{thm:l0dDhard1} to any graph $H$, $\gamma(C) = 0$, and $\binom{l}{2} < |C|$, then $\beta(C) \ge l + 1$.
    \label{claim:l0widecluster}
\end{claim}
\begin{proof}
The proof is by induction on $l$. The base is $l = 1$, and each non-empty cluster contains at least one vector and so at least 2 coordinates with vertices, we assume $\binom{1}{2} = 0$.

For the general case, if there are at least $l$ occurences of a vertex $v$ in a coordinate $i$, then there are at least $(l + 1)$ coordinates with vertices. Each vector with $v$ in the $i$-th coordinate has also some other vertex in some other coordinate. As in each coordinate all vertex values are equal, it could not be that two of the vectors with the value $v$ in the $i$-th coordinate share the second vertex-valued coordinate, since then they would represent the same edge.

So each coordinate has at most $(l - 1)$ vertex occurences, otherwise the claim holds. Select a coordinate $j$ which contains some vertex value $u$ and remove the $j$-th coordinate and all vectors which have the value $u$ in the $j$-th coordinate. That corresponds to the natural restriction $C'$ of the cluster $C$ to a subgraph $H - u$. The size of $C'$ is at least $\binom{l}{2} + 1 - (l - 1) = \binom{l - 1}{2} + 1$, and by induction there are at least $l$ coordinates which contain vertex values, so the original cluster $C$ has at least $l + 1$ such coordinates, since there is also the $j$-th coordinate with the vertex value $u$.
\end{proof}

Now consider a cluster $C$ with $\gamma(C) = 0$. Let $l$ be the largest value with $\binom{l}{2} + 1 \le |C|$, so $|C| \le \binom{l + 1}{2}$. Since $|C| \le \binom{k}{2}$, $l + 1 \le k$. By Claim \ref{claim:l0widecluster}, $\beta(C) \ge l + 1$, then 
\[\frac{\beta(C) - 2}{|C| - 1} \ge \frac{l - 1}{\binom{l + 1}{2} - 1} = \frac{2}{l + 2} \ge \frac{2}{k + 1} = \kappa,\]
and so if $l + 1 < k$, the inequality is strict. It is also strict if $l + 1 = k$ and $|C| < \binom{k}{2}$, as the denominator becomes larger in the first step. Thus the only possibility of getting exactly $\kappa$ is when $|C| = \binom{k}{2}$.

But then we have exactly $k \cdot (k - 1)$ vertex values across $k$ coordinates, and each coordinate has at most $(k - 1)$ vertex values by the argument in Claim \ref{claim:l0widecluster}, so each coordinate must have exactly $(k - 1)$ vertex values. Since $\gamma(C) = 0$, they must be all equal. Denote the common vertex value in the $i$-th coordinate as $v_i$. Since each occurence of $v_i$ in the $i$-th coordinate corresponds to an edge to a different $v_j$, vertices $v_1$, \dots, $v_k$ form a clique in $G$.
    
In the case $\gamma(C) > 0$, consider a new cluster $C'$ which is obtained from $C$ by removing all vectors which have a vertex-valued entry not equal to the centroid value. Assume for now that $|C'| \ge 2$. By the proof above, $\frac{\beta(C') - 2}{|C'| - 1} \ge \kappa$, since $\gamma(C') = 0$.
The value $\frac{\beta(C) - 2 + \gamma(C)}{|C| - 1}$ could be obtained from $\frac{\beta(C') - 2}{|C'| - 1}$ 
by adding $\gamma(C) + (\beta(C) - \beta(C')$ to the numerator and $|C| - |C'|$ to the denominator. Removing vectors could not increase $\beta$, so $\beta(C) - \beta(C') \ge 0$, and $\gamma(C) \ge |C| - |C'|$ since each of the removed vectors has at least one vertex value not equal to the centroid value. If $\frac{\beta(C') - 2}{|C'| - 1} \ge 1$, then the new fraction is also at least 1 and so striclty greater than $\kappa$. If $|C'| \le 1$, then $\frac{\beta(C) - 2 + \gamma(C)}{|C| - 1} \ge 1$ since $\beta(C) \ge 2$ and $\gamma(C) \ge |C| - |C'|$. If $\frac{\beta(C') - 2}{|C'| - 1} < 1$, then the new fraction became strictly larger, and so stricly larger than $\kappa$. In all cases, the inequality is strict when $\gamma(C) > 0$.

\end{proof}

Now to \probClustSelect: the reduction is almost the same, only we start from \probMultiClique, and for each pair of indices $\{i, j\}$, $1 \le i < j \le k$ we obtain the set of vectors $X_{i, j}$ from edges in $G$ starting in color $i$ and ending in color $j$. The vectors are constructed in the same way as in the previous reduction. All weights are set to one. The value of $D$ is the same, $D = \binom{k}{2}(k - 2)$.

Since vectors are constructed in the same way, all statements about the cost of grouping them remain valid, in particular Lemma \ref{lemma:l0clustercost}. Only now the statement of \probClustSelect already guarantees that we select exactly one cluster and exactly one vector from each $X_{i, j}$, so exactly one edge between each pair of colors. And by Lemma \ref{lemma:l0clustercost} only the proper $k$-clique has the optimal cost.

%Finally, we prove that the stated ETH bounds hold. One of the reductions above together with an $n^{o(d + D^{1/3})}$ algorithm for \probClust or an $n^{o(d + t^{1/2} + D^{1/3})}$ algorithm for \probClustSelect would give a $|V(G)|^{o(k)}$ algorithm for \probClique since $n = \binom{k}{2} |E(G)|$, $d = k$, $t = \binom{k}{2}$ and $D = \binom{k}{2}(k-2)$.

\end{proof}

Note that \probClustSelect with the $L_0$ distance is very similar to the known problem \textsc{Consensus String With Outliers}, studied e.g. in~\cite{BoucherLL11}. The only difference of \probClustSelect is that we have to select one point from each of the given sets, whereas in \textsc{Consensus String With Outliers} the goal is to select the arbitrary subset of size $(n - k)$. The construction from Theorem~\ref{thm:l0dDhard1} also shows W[1]-hardness of \textsc{Consensus String With Outliers} with respect to $(d + D + n - k)$ in the case of unbounded alphabet.

%!TEX root = Integer_clustering.tex
\section{The $L_\infty$ distance}\label{sec:tinfty}

In this section, we consider the case $p = \infty$. We prove two hardness results of \probClust: \classW1-hardness when parameterized by $D$ and \classNP-hardness in the case $k = 2$.

First, we prove some useful facts about the structure of optimal cluster centroids. The one thing, in which the $L_\infty$ distance is harder than all other distances in our consideration, is that even when the cluster is given, we can not just find the optimal cluster centroid by optimizing the value in each coordinate independently. So there seems to be no simple rule of finding the optimal cluster centroid of a given cluster. However, one could still do that in polynomial time by solving a linear program.

\begin{claim}
Given a multiset $C$ of vectors in $\mathbb{Z}^d$, there is a polynomial time algorithm to find  $c \in \mathbb{R}^d$ minimizing
\[\sum_{x \in C} \distinfty(x, c).\]
%in polynomial time.
\end{claim}
\begin{proof}
    We reduce to solving a linear program, which we define next. Denote $C = \{x_1, \dots, x_n\}$, introduce variables $c_1$, \dots, $c_d$ corresponding to coordinates of the cluster centroid and variables $d_1$, \dots, $d_n$, where $d_i$ corresponds to the value $\distinfty(x_i, c)$. The following linear program solves to the minimum total distance.
\begin{gather*}
\sum_{i=1}^n d_i \to \min\\
x_i[j] - c_j \le d_i \quad \forall \, i, j: 1 \le i \le n, 1 \le j \le d\\
c_j - x_i[j] \le d_i \quad \forall \, i, j: 1 \le i \le n, 1 \le j \le d
\end{gather*}

\end{proof}

The next claim shows that we could only consider half-integral cluster centroids.

\begin{claim}
For any multiset $C$ of vectors in $\mathbb{Z}^d$, the vector $c \in \mathbb{R}^d$ which minimizes
\[\sum_{x \in C} \distinfty(x, c)\]
could always be chosen from $\frac{1}{2}\mathbb{Z}^d$ (coordinates are either integer or half-integer).
\label{claim:linfhalf}
\end{claim}
\begin{proof}
Assume that we have an optimal solution $c$ which has at least one coordinate not of the form $z/2$, $z\in\mathbb{Z}$.
For $a \in \mathbb{R}$ we denote $\text{frac}(a) = a - \lfloor a \rfloor$, and
\[\text{rem}(a) = \begin{cases}\text{frac}(a), \quad \text{if} \text{ frac} (a) < 1/2\\
       1 - \text{frac}(a), \quad \text{otherwise}\end{cases},\]
calling this value the \emph{remainder} of $a$.

%Choose any such coordinate $j$ and denote $\xi = c_j - \lfloor c_j \rfloor$ if it's less than half or otherwise $\xi = \lceil c_j \rceil - c_j$. Note that $\xi < 1/2$, and so $\xi \ne 1 - \xi$.

We could partition all coordinates on equivalence classes by remainder of $c$. One could also define a partition of all vectors by the remainder of the distance to $c$. These two partitions are related in the following sense:
if $\distinfty(x, c)$ has remainder $\xi$ then each coordinate $j$ where $|x[j] - c[j]| = \distinfty(x, c)$ also has remainder $\xi$, and vice versa. Now we take one particular remainder and show that we can shift it without losing optimality.

There are two kinds of vectors with the particular remainder $\xi$: call \emph{bottom} those vectors $x$ for which $\text{frac}(\distinfty(x, c)) = \xi$, and call \emph{top} those vectors $x$ for which $\text{frac} (\distinfty(x, c)) = 1 - \xi$. Similarly, there are also two kinds of
coordinates of $c$, which we also call bottom and top depending of the value of $\text{frac}(c[j])$.

Consider a bottom cordinate $j$. Increasing $c[j]$ increases $|x[j] - c[j]|$ for all bottom vectors $x$, and decreases $|x[j] - c[j]|$ for all top vectors $x$. Decreasing $c[j]$ does the opposite, as well as increasing a top coordinate.
So if we take some sufficiently small value $\beta$ and simultaneously increase all bottom coordinates and decrease all top coordinates by $\beta$ then for all bottom vectors their distance will become larger by $\beta$, and for all top vectors --- smaller by $\beta$.
An if we do the opposite, the bottom vectors will cost less and the top vectors will cost more. Then, we could just take the group which has more vectors (bottom or top) and choose that action which decreases the distance for these vectors. The larger group has at least as many vectors as the smaller group, so the total distance does not increase.

It remains to see which value of $\beta$ we could take. We could safely shift until we either reach a value in $\frac{1}{2}\mathbb{Z}$ or another remainder. In any case, we reduce the number of distinct remainders by one, and so we conclude the proof by doing this inductively over the number of distinct remainders.

\end{proof}

By Claim \ref{claim:linfhalf}, the $\alpha$-property holds with $\alpha = 1/2$, since at most one vector could be equal to the cluster centroid, and all others have distance at least $1/2$ due to half-integrality. We can also see that when the problem is parameterized by $d + D$, it is FPT.

\begin{claim}
    \probClust with the $L_\infty$ distance is FPT when parameterized by $d + D$.
\end{claim}
\begin{proof}
    We use Theorem \ref{thm:genfpt}. We have the $\alpha$-property, and for the set $\mathcal{D}$ of all possible cluster costs not exceeding $D$ we could take all half-integral values not exceeding $D$ by Claim \ref{claim:linfhalf}. All that remains is to solve \probClustSelect in FPT time.

    For that, we try all possible $x_1 \in X_1$, and then try each possible resulting cluster centroid $c$. Since $\distinfty(x_1, c) \le D$ and $c$ is half-integral by Claim \ref{claim:linfhalf}, we can try only vectors $c$ of this form, and that is done in time $(2D + 1)^d$.
\end{proof}

\subsection{\classW1-hardness when parameterized by $D$}

Knowing that \probClust with the $L_\infty$ distance is FPT when parameterized by $d + D$, the next natural question is, is the problem FPT or \classW1-hard when parameterized only by $D$?
We show that \classW1-hardness is the case, proving Theorem \ref{thm:linftydDhard1}, which we recall here for convenience.

\linftyhard*
\begin{proof}

    First, we show a reduction from \probClique to \probClust. Given a graph $G$ and a clique size $k$, we construct the following instance of the clustering problem.

    We set the dimension to $|V(G)| + \binom{|V(G)|}{2} - |E(G)|$. We take $|V(G)|$ vectors $\{x_i\}_{i=1}^{|V(G)|}$ corresponding to vertices. For the vertex $v$, first $|V(G)|$ coordinates are set to zero, except $v$-th coordinate, which is set to 2.

    The last $\binom{|V(G)|}{2} - |E(G)|$ coordinates correspond to non-edges, vertex pairs which are not connected by an edge. For each vertex pair $\{u, v\} \notin E(G)$ in the coordinate $\{u, v\}$ we set $x_u$ to $2$, $x_v$ to $-2$, the order on $u$, $v$ is chosen arbitrarily, and all other vectors to zero.

    Finally, we set the number of clusters to $|V(G)| - k + 1$ and the total distance to $k$. We show an example on how the reduction works in Figure \ref{fig:linftydDhard1}.

        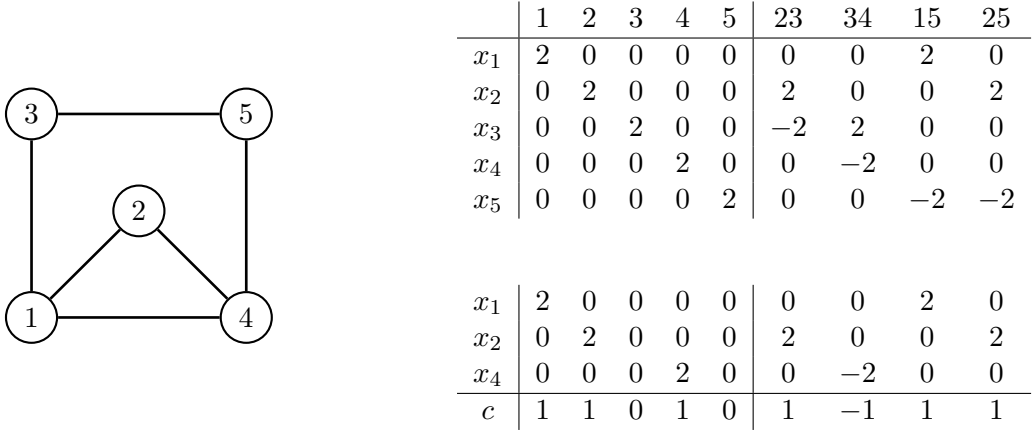
\begin{figure}[ht]
            \centering
            \begin{subfigure}{0.4\textwidth}
                \centering

                \begin{tikzpicture}[auto,  node distance=2cm,every loop/.style={},thick,main node/.style={circle,draw}]
                    \node[main node] at (0,0) (1) {1};
                    \node[main node] (2) [above right of=1] {2};
                    \node[main node] (3) [above =2cm of 1] {3};
                    \node[main node] (4) [below right of=2] {4};
                    \node[main node] (5) [above =2cm of 4] {5};
                    \draw[line width=1pt]
                    (1) -- (2);
                    \draw[line width=1pt]
                    (1) -- (3);
                    \draw[line width=1pt]
                    (1) -- (4);
                    \draw[line width=1pt]
                    (2) -- (4);
                    \draw[line width=1pt]
                    (3) -- (5);
                    \draw[line width=1pt]
                    (5) -- (4);
                \end{tikzpicture}
            \end{subfigure}\hfill
            \begin{subfigure}{0.6\textwidth}
                \centering
                \begin{tabular}{c|c c c c c|c c c c}
                    &1&2&3&4&5&23&34&15&25\\\hline
                    $x_1$&2&0&0&0&0&0&0&2&0\\
                    $x_2$&0&2&0&0&0&2&0&0&2\\
                    $x_3$&0&0&2&0&0&$-2$&2&0&0\\
                    $x_4$&0&0&0&2&0&0&$-2$&0&0\\
                    $x_5$&0&0&0&0&2&0&0&$-2$&$-2$\\
                \end{tabular}
                \par\bigskip
                \par\bigskip
                \begin{tabular}{c|c c c c c|c c c c}
                    $x_1$&2&0&0&0&0&0&0&2&0\\
                    $x_2$&0&2&0&0&0&2&0&0&2\\
                    $x_4$&0&0&0&2&0&\makebox[\widthof{$-2$}][c]{0}&$-2$&\makebox[\widthof{$-2$}][c]{0}&\makebox[\widthof{$-2$}][c]{0}\\
                    \hline
                    $c$&1&1&0&1&0&1&$-1$&1&1\\
                \end{tabular}
            \end{subfigure}

            \caption{An example illustrating the reduction in Theorem \ref{thm:linftydDhard1}: an input graph $G$, the vectors produced by the reduction (for clarity, the coordinates corresponding to vertices and to non-edges are separated), and the only composite cluster in the resulting optimal clustering of cost 3, corresponding to the clique on $\{1, 2, 4\}$. Note that $\distinfty(x_1, c) = \distinfty(x_2, c) = \distinfty(x_4, c) = 1$.}
            \label{fig:linftydDhard1}
        \end{figure}

    If there is a clique of size $k$ in $G$, then we have a solution of cost $k$: take $k$ vectors corresponding to the clique vertices in one cluster, and make all other clusters trivial. For the only nontrivial cluster $C$, we can always choose $c$ so that $|x[j] - c[j]| \le 1$ for any $x \in C$ and for any coordinate $j$. Each vertex coordinate has only 0 and $2$, so setting $c$ to 1 there suffices.
As in $C$ we have an edge between any two vertices, in any non-edge coordinate $j$ there are either all zeroes, or zeroes and $2$, or zeroes and $-2$. 
In each of the cases there is a suitable value for $c_j$: $0$, $1$ or $-1$ correspondingly.

Next, we prove that any solution has cost at least $k$, and any solution which is not a $k$-clique has stricly larger cost. For that, we prove the following claim.

\begin{claim}
    In the instance above, the cost of any cluster $C$ containing at least two vectors is at least $|C|$. If there is at least one non-edge in $C$, then the cost is at least $|C| + 1$.
\label{claim:linfopt}
\end{claim}
\begin{proof}
    Denote the cluster centroid as $c$. If each vector $x$ in $C$ has $\distinfty(x, c) \ge 1$, the first statement is trivial. So assume that there is a vector $x^*$ in $C$ such that $\distinfty(x^*, c) = \xi < 1$. Consider the coordinate $j^*$ which corresponds to the same vertex as the vector $x^*$, $x^*[j^*] = 2$, and all other vectors are zero in the coordinate $j^*$. As $\distinfty(x^*, c) = \xi$, $c[j^*] \ge 2 - \xi$. Then, for any other $x \in C$, $\distinfty(x, c) \ge 2 - \xi > 1$. The total cost of the cluster is at least $\xi + (|C| - 1) (2 - \xi) = 2 + (|C| - 2)(2 - \xi) \ge |C|$, as $2 - \xi > 1$.

    Now to the second part of the claim. Assume there are only two vectors in $C$ and they do not have an edge, there is a coordinate $j^*$ where one is 2 and the other is $-2$. No matter what we choose for $c[j^*]$, the cost is at least $|2 - c[j^*]| + |-2 - c[j^*]| \ge 4$, and the statement follows. So assume that $|C| \ge 3$ and there is a coordinate $j^*$ corresponding to a non-edge in $C$. One vector from $C$ has 2 in the coordinate $j^*$, another $-2$, and all others have 0. Then there is a vector in $C$ with distance to $c$ of at least 2, as either $c[j^*] \ge 0$ and $|-2 - c[j^*]| \ge 2$ or $c[j^*] < 0$ and $|2 - c[j^*]| > 2$. Let us just forget about this vector and consider all other vectors in $C$. There are $|C| - 1 \ge 2$ of them, and by the reasoning in the proof of the first statement, their cost is at least $|C| - 1$. In this proof we considered only vertex coordinates, so the vector we forgot and the $j^*$-th coordinate (which is a non-edge coordinate) does not affect it. So, the total cost is at least $|C| - 1 + 2 = |C| + 1$.
\end{proof}

Assume that we have $l \ge 1$ nontrivial clusters of sizes $\{t_i\}_{i=1}^l$, nontrivial means that the size is at least two, $t_i \ge 2$ for $i \in \{1, \dots, l\}$. By Claim \ref{claim:linfopt}, the total cost is at least
\[\sum_{i=1}^l t_i = k + l - 1 \ge k,\]
as there are $|V(G)| - k + 1$ clusters in total, $|V(G)| - k + 1 - l$ trivial clusters, and the total number of vectors is $|V(G)| = \sum_{i=1}^l t_i + |V(G)| - k + 1 - l$, from which it follows that $\sum_{i = 1}^l t_i = k + l - 1$. So no solution has cost less than $k$.

Also, if there are at least two nontrivial clusters, then $k + l - 1 \ge k + 1$. So if a solution has cost $k$, it must have only one nontrivial cluster, and its size must be $k$.

Finally, assume that the solution indeed has only one nontrivial cluster, but there is a non-edge in it. Then, as the size is $k$,  by Claim \ref{claim:linfopt} its cost is at least $k + 1$. So only a $k$-clique has cost $k$, which proves the correctness of the reduction.

Now, to \probClustSelect. We consider essentially the same reduction, only we start from \probMultiClique. We obtain sets of vectors $X_1$, \dots, $X_k$ in the same way as $X$ in the reduction above, only vectors obtained from vertices of color $j$ are put into $X_j$. The total distance parameter is also set to $k$. So parameters $t$ and $D$ of the obtained instance have the same value as the starting parameter $k$.

Since vectors are constructed in the same way, Claim \ref{claim:linfopt} still works. And now the statement of \probClustSelect enforces that exactly one cluster of $k$ vectors is selected. By Claim \ref{claim:linfopt} it could be done with the cost $k$ if and only if there is a colorful $k$-clique in the original graph.

%Finally, we prove that the stated ETH bounds hold. An $n^{o(D)}$ algorithm for \probClust together with the first reduction would give a $|V(G)|^{o(k)}$ algorithm for \probClique. The same is with an $n^{o(d + t)}$ algorithm for \probClustSelect together with the second reduction. Both of these hold since $n = |V(G)|$, $t = k$ and $D = k$.
\end{proof}

\subsection{\classNP-hardness when $k = 2$}

In this subsection we prove \classNP-hardness of \probClust with the $L_\infty$ distance when $k = 2$. Intuitively, if we consider the previous reduction, partitioning the vectors optimally into two clusters loosely corresponds to partitioning the vertices into two sets such that there are as many as possible vertices having no edges inside their set. Which, in turn, is \probOCT: the problem of removing the smallest number of vertices so that the remaining graph is biparite. However, to make everything really work, we need to consider a modified version of \probOCT which we call \probHIOCT.

\defproblema{\probHIOCT}%
{An undirected graph $G$, an integer $t$.}%
{Is there an assignment $\delta:V(G) \to \{0, 1, 2\}$, such that $\sum_{v \in V(G)} \delta(v) \le t$ and $G - S$ is bipartite, where $S = \{\{u, v\} \in E(G) : \delta(u) + \delta(v) \ge 2\}$?}

First we show that \probHIOCT is also \classNP-hard by constructing a reduction from \probSAT.

\begin{lemma}
    There is a polynomial time reduction from \probSAT to \probHIOCT.
    \label{lemma:hioct}
\end{lemma}
\begin{proof}
    Given an instance of \probSAT with $n$ variables and $m$ clauses, make a graph $G$ as follows. The example of the reduction is given in Figure \ref{fig:hioct}. For each variable $x_i$, introduce two vertices $x_i$ and $x_i'$, connect them with an edge. Also introduce $2n + 1$ vertices $y_{i, j}$ connect them to both $x_i$ and $x_i'$. 
    
    For each clause $C_j$ introduce four vertices $C_{j,1}$,\dots,$C_{j,4}$. Consider following seven vertices: $C_{j, 1}$, \dots, $C_{j, 4}$, and three variable vertices which are present in $C_j$: if $x_i \in C_j$ then we consider the vertex $x_i$, and if $\neg x_i \in C_j$ then we consider the vertex $x_i'$. Connect all these seven vertices in a cycle such that each variable vertex is adjacent to two clause vertices. Finally, set $t$ to $2n$.

    \begin{figure}[ht]
        \centering
        \begin{tikzpicture}[auto,  node distance=2cm,every loop/.style={},thick,main node/.style={circle,draw, inner sep=0cm, minimum width=.75cm}]
            \node[main node] (x1) {$x_1$};
            \node[main node] (x1p) [right of=x1] {$x_1'$};
            \node[main node] (y11) [above left =1cm and .5cm of x1] {$y_{1, 1}$};
            \node[main node] (y17) [above right =1cm and .5cm of x1p] {$y_{1, 7}$};
%            \node[draw=none] (y1m) [right=2cm of y11] {$\cdots$};

            \draw[line width=1pt]
            (x1) -- (x1p)
            (y11) -- (x1)
            (y11) -- (x1p)
            (y17) -- (x1)
            (y17) -- (x1p);
            \path (y11.base) -- node[anchor=center] (y1m) {{\Large $\cdots$}} (y17.base);

            \node[main node] (x2) [right=3cm of x1p] {$x_2$};
            \node[main node] (x2p) [right of=x2] {$x_2'$};
            \node[main node] (y21) [above left =1cm and .5cm of x2] {$y_{2, 1}$};
            \node[main node] (y27) [above right =1cm and .5cm of x2p] {$y_{2, 7}$};
%            \node[draw=none] (y1m) [right=2cm of y11] {$\cdots$};

            \draw[line width=1pt]
            (x2) -- (x2p)
            (y21) -- (x2)
            (y21) -- (x2p)
            (y27) -- (x2)
            (y27) -- (x2p);
            \path (y21.base) -- node[anchor=center] (y2m) {{\Large $\cdots$}} (y27.base);

            \node[main node] (x3) [right=3cm of x2p] {$x_3$};
            \node[main node] (x3p) [right of=x3] {$x_3'$};
            \node[main node] (y31) [above left =1cm and .5cm of x3] {$y_{3, 1}$};
            \node[main node] (y37) [above right =1cm and .5cm of x3p] {$y_{3, 7}$};
%            \node[draw=none] (y1m) [right=2cm of y11] {$\cdots$};

            \draw[line width=1pt]
            (x3) -- (x3p)
            (y31) -- (x3)
            (y31) -- (x3p)
            (y37) -- (x3)
            (y37) -- (x3p);
            \path (y31.base) -- node[anchor=center] (y3m) {{\Large $\cdots$}} (y37.base);

            \node[main node] (c11) [below right = 1cm and 1.3cm of x1] {$C_{1, 1}$};
            \node[main node] (c12) [right = 2.5cm of c11] {$C_{1, 2}$};
            \node[main node] (c13) [right = 2.5cm of c12] {$C_{1, 3}$};
            \node[main node] (c14) [right = 2.5cm of c13] {$C_{1, 4}$};
            \draw[line width=1pt,color=blue]
            (c11) -- (x1)
            (x1) -- (c12)
            (c12) -- (x2p)
            (x2p) -- (c13)
            (c13) -- (x3)
            (x3) -- (c14);
            \path[line width=1pt,color=blue]
            (c11) edge[bend right] (c14);
        \end{tikzpicture}
        \caption{A graph obtained from the 3CNF-formula $(x_1 \lor \lnot x_2\lor x_3)$ by the reduction from Lemma \ref{lemma:hioct}. A 7-cycle corresponding to the only clause of the formula is highlighted in blue.}
        \label{fig:hioct}
    \end{figure}
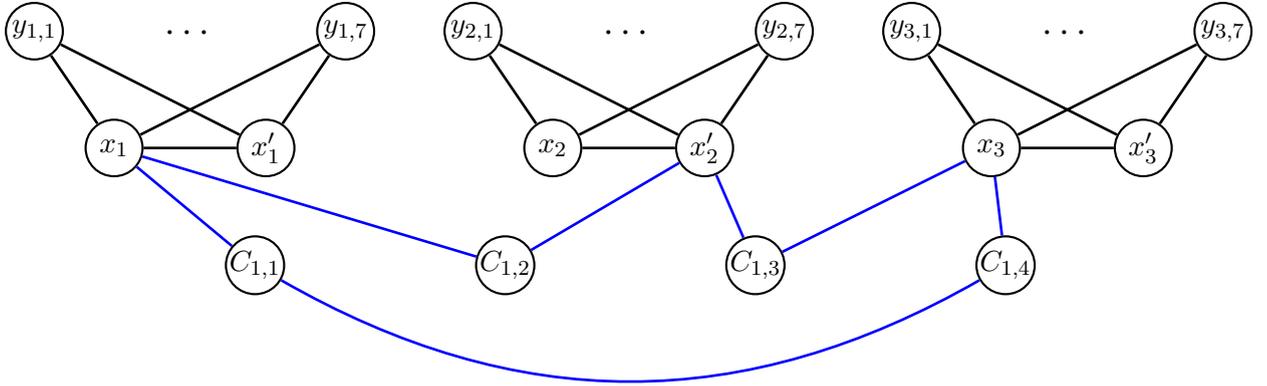

    First, assume there is a satisfying assignment. Consider the following $\delta : V(G) \to \{0, 1, 2\}$: if $x_i$ is true, $\delta(x_i) = 2$, otherwise $\delta(x_i') = 2$, on all other vertices $\delta \equiv 0$. Clearly, $\sum_{v \in V(G)} \delta(v) = 2n$.
    
    Since $\delta$ does not take value $1$, deleting edges $\{u, v\}$ with $\delta(u) + \delta(v) \ge 2$ is equivalent to deleting vertices on which $\delta$ is 2. From each vertex gadget we deleted either $x_i$ or $x_i'$, so the remaining part is a star with leaves $y_{i, j}$ and center $x_i$ or $x_i'$. Since the assignment we started from is satisfying, from each clause cycle we deleted at least one vertex. So each cycle present in $G$ lost at least one vertex, and what remains is bipartite.

    Now assume there is a solution $\delta$ to the \probHIOCT instance. We claim that $\delta(x_i) + \delta(x_i') \ge 2$ for each variable $x_i$. Consider a 2-coloring of $G - S$: either $x_i$ and $x_i'$ have the same color or not. In the former case, $\delta(x_i) + \delta(x_i') \ge 2$ since the edge $\{x_i, x_i'\}$ must be removed.

    If $x_i$ and $x_i'$ have different colors, assume that $\delta(x_i) \le 1$ and $\delta(x_i') \le 1$. Then, each of the $2n + 1$ vertices $y_{i, j}$ takes one of the two colors, and so has an incident edge to $x_i$ or $x_i'$ which needs to be deleted. But then, $\delta(y_{i, j}) \ge 1$ for each $j$, and the total cost on these vertices is already $2n + 1$. Then either $\delta(x_i) = 2$ or $\delta(x_i') = 2$.

    So we have $n$ variables and $\delta$ is at least $2$ on each pair of variable vertices, and in total $\delta$ is at most $2n$. Then $\delta$ has to be exactly $2$ on each variable pair, and zero on all other vertices. Now we claim that on each clause cycle there is a variable vertex $v$ with $\delta(v) = 2$. If not, then none of the cycle edges gets deleted, as $\delta$ is equal to zero on clause vertices. But then the remaining graph could not be bipartite, since it contains an odd cycle.

    To get a satisfying assignment, set $x_i$ to true if $\delta(x_i) = 2$, or to false otherwise. In particular, if $\delta(x_i') = 2$, $x_i$ is set to false, since $\delta(x_1) + \delta(x_1') = 2$. Each clause is satisfied since each clause cycle contains a variable vertex on which $\delta$ is equal to $2$.
\end{proof}

Now we prove \classNP-hardness of \probClust with $p=\infty$ and $k=2$ by constructing a reduction from \probHIOCT.

\begin{theorem}
    \probClust with distance $\distinfty$ is \classNP--hard when $k=2$.
    \label{thm:linfoct}
\end{theorem}
\begin{proof}
    Consider an instance $(G, t)$ of \probHIOCT, if $t \ge |V(G)|$, we have a \yesinstance since $\delta \equiv 1$ deletes all edges from the graph, so we may assume $t < |V(G)|$. Remove all isolated vertices in $G$ and add $t + 5$ isolated edges to $G$, it clearly does not change the type of the instance.
    The number of clusters $k$ is $2$, set the dimension $d$ to $|E(G)|$, each coordinate corresponds to an edge. For each vertex $v \in V(G)$ add a vector $x_v$ to $X$ with all coordinates set to zero. Then, for each edge $\{u, v\} \in E(G)$ set $x_u[u, v]$ to $2$ and $x_v[u, v]$ to $-2$, the order on $u, v$ is chosen arbitrarily. Finally, set $D$ to $|V(G)| + t$. An example is given in Figure \ref{fig:linfoct}, additional isolated edges are dropped out for clarity.
    \begin{figure}[ht]
        \centering
        \begin{subfigure}{.45\textwidth}
            \centering
            \subcaption{A starting graph $G$, $t = 2$.}
        \begin{tikzpicture}[auto,  node distance=2cm,every loop/.style={},thick,main node/.style={circle,draw, inner sep=0cm, minimum width=.75cm}]
            \node[main node] (1) {1};
            \node[main node] (2) [right of=1]{2};
            \node[main node] (3) [below of=1]{3};
            \node[main node] (4) [right of=3]{4};

            \draw[line width=1pt]
            (1) -- (3)
            (1) -- (4)
            (1) -- (2)
            (2) -- (3)
            (2) -- (4);
        \end{tikzpicture}
        \end{subfigure}\hfill
        \begin{subfigure}{.45\textwidth}
            \centering
            \subcaption{The obtained instance: set of vectors $X = \{x_1, x_2, x_3, x_4\}$, $D = 6$.}
        \begin{tabular}{c c c c c c c}
            edges:&12&13&14&23&24&\\
        \hline
            $x_1=($&2&2&2&0&0&$)$\\
            $x_2=($&$-2$&0&0&2&2&$)$\\
            $x_3=($&0&$-2$&0&$-2$&0&$)$\\
            $x_4=($&0&0&$-2$&0&$-2$&$)$
        \end{tabular}
        \end{subfigure}
        \par\bigskip
        \begin{subfigure}{.42\textwidth}
            \centering
            \subcaption{A possible solution: $\delta(1) = \delta(3) =\delta(4) = 0$, $\delta(2) = 2$. Edges from $S$ are dashed, a 2-coloring of $G - S$ is in red and blue.}
        \begin{tikzpicture}[auto,  node distance=2cm,every loop/.style={},thick,main node/.style={circle,draw, inner sep=0cm, minimum width=.75cm}]
            \node[main node,color=red] (1) {1};
            \node[main node,color=red] (2) [right of=1]{2};
            \node[main node,color=blue] (3) [below of=1]{3};
            \node[main node,color=blue] (4) [right of=3]{4};

            \draw[line width=1pt]
            (1) -- (3)
            (1) -- (4);
            \draw[line width=1pt,dashed]
            (1) -- (2)
            (2) -- (3)
            (2) -- (4);
        \end{tikzpicture}
        \end{subfigure}\hfill
        \begin{subfigure}{.52\textwidth}
            \centering
            \subcaption{The corresponding clustering of cost 6, $C_1 = \{x_1, x_2\}$, $C_2 = \{x_3, x_4\}$, and optimal centroids $c_1$, $c_2$.}
        \begin{tabular}{c c}
            $\textcolor{red}{c_1}=(\phantom{-}1,\phantom{-}1,\phantom{-}1,\phantom{-}1,\phantom{-}1)$\phantom{,}&\\
            $\textcolor{red}{x_1}=(\phantom{-}2,\phantom{-}2,\phantom{-}2,\phantom{-}0,\phantom{-}0)$,&$\distinfty(x_1, c_1) = 1$\\
            $\textcolor{red}{x_2}=(-2,\phantom{-}0,\phantom{-}0,\phantom{-}2,\phantom{-}2)$,&$\distinfty(x_2, c_1) = 3$\\
            &\\
            $\textcolor{blue}{c_2}=(\phantom{-}0,-1,-1,-1,-1)\phantom{,}$&\\
            $\textcolor{blue}{x_3}=(\phantom{-}0,-2,\phantom{-}0,-2,\phantom{-}0),$&$\distinfty(x_3, c_1) = 1$\\
            $\textcolor{blue}{x_4}=(\phantom{-}0,\phantom{-}0,-2,\phantom{-}0,-2),$&$\distinfty(x_4, c_1) = 1$
        \end{tabular}
%        \begin{tabular}{c c c c c c c c}
%            $\textcolor{red}{c_1}=($&1&1&1&1&1&$)$\phantom{,}&\\
%            $\textcolor{red}{x_1}=($&2&2&2&0&0&$)$,&$\distinfty(x_1, c_1) = 1$\\
%            $\textcolor{red}{x_2}=($&-2&0&0&2&2&$),$&$\distinfty(x_2, c_1) = 3$\\
%            &&&&&&\\
%            $\textcolor{blue}{c_2}=($&0&-1&-1&-1&-1&$)\phantom{,}$&\\
%            $\textcolor{blue}{x_3}=($&0&-2&0&-2&0&$),$&$\distinfty(x_3, c_1) = 1$\\
%            $\textcolor{blue}{x_4}=($&0&0&-2&0&-2&$),$&$\distinfty(x_4, c_1) = 1$
%        \end{tabular}
        \end{subfigure}
        \caption{An illustration of the reduction from Theorem \ref{thm:linfoct}.}
        \label{fig:linfoct}
    \end{figure}
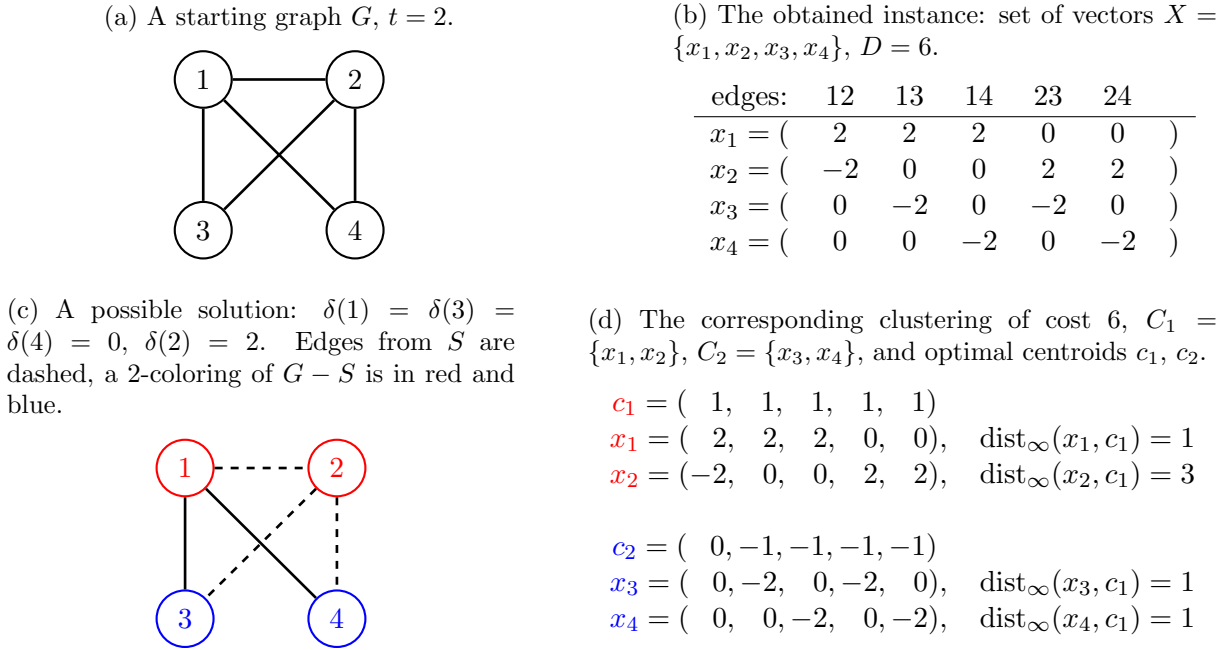

    If $(G, t)$ is a \yesinstance of \probHIOCT, consider the solution $\delta$. Split vectors into clusters according to any proper 2-coloring of $G - S$. Now we show the way to select cluster centroids so that each vertex $v$ has distance at most $1 + \delta(v)$ to the corresponding centroid. We consider separately each of two clusters and each coordinate, indexed by an edge $\{u, v\} \in E(G)$. For a cluster $C$, there are three cases on how $u$ and $v$ are present in the cluster, for each of them we assign a particular value to the cluster centroid $c$ in the coordinate $\{u, v\}$.
    \begin{itemize}
        \item If $u$ and $v$ are both not in $C$, for vectors in $C$ all entries in the coordinate $\{u, v\}$ are zero, and we set $c[u, v]$ also to zero. Each vector is at distance zero to the centroid in this coordinate.
        \item If only one of $u$ and $v$ are in $C$, for vectors in $C$ all entries in the corresponding coordinate are zero, except one entry corresponding to the edge's endpoint belonging to $C$, which is either $2$ or $-2$. Set $c[u, v]$ to $1$ or $-1$, correspondingly, then each vector is at distance $1$ in this coordinate.
        \item If both $u$ and $v$ are in $C$, w.l.o.g $x_u[u, v]$ is $2$ and $x_v[u, v]$ is $-2$, and all other points are zero. It must hold that $\delta(u) + \delta(v) \ge 2$, either $\delta(u) = \delta(v) = 1$ or w.l.o.g $\delta(u) = 2$ and $\delta(v) = 0$. 
            In the former case, set $c[u, v]$ to zero, then all vectors have distance zero, $x_u$ and $x_v$ have distance $2$ in this coordinate. In the latter case, set $c[u, v]$ to $-1$, then $u$ is at distance $3$, and all other vectors, including $v$, are at distance $1$.
    \end{itemize}

    For any $v \in V(G)$, since it holds for all coordinates that distance from $x_v$ to the corresponding cluster centroid is at most $1 + \delta(v)$, then the $L_\infty$ distance is also at most $1 + \delta(v)$, and the total cost of the clustering defined above is at most
\[\sum_{v \in V(G)} 1 + \delta(v) = |V(G)| + t.\]

    In the other direction, assume there is a clustering $C_1$, $C_2$ with centroids $c_1$, $c_2$ such that the total cost is at most $D$. By Claim \ref{claim:linfhalf} we may assume that centroids are integral, and for any vector the distance to the nearest centroid is also an integer. We also may assume that centroids are between $-2$ and $2$ in each coordinate since all the input vectors have entries in this range, and so we could move the centroids to the same range without increasing distances.

    So, each vector has distance in $\{0, 1, 2, 3, 4\}$ to the closest centroid. We claim that it could not be that a vector $x_v$ has distance zero: in this case w.l.o.g $x_v = c_1$, and so $c_1$ is equal to $2$ or $-2$ in some coordinate, since each vertex has at least one incident edge. But then each vector in $C_1$ has distance at least $2$ to $c_1$. And since at most two vectors could be equal to the centroids, each of the remaining $|V(G)| - 2$ vectors has distance at least 1. Consider $t + 5$ isolated edges, at least $t + 3$ of them do not have
    any endpoint equal to one of $c_1$ and $c_2$. For these edges, the total distance of their endpoints is at least $3$: either their endpoints are in different clusters, and so the endpoint in $C_1$ costs at least $2$, or both endpoints are in the same cluster, and in total they cost $4$ since there are simultaneously values $2$ and $-2$ in the coordinate corresponding to this edge. So each of the $t + 3$ edges increases the cost by additional $1$, and the total cost is at least $|V(G)| - 2 + t + 3 > |V(G)| + t$.

    Since each vector has distance at least $1$, we may assume that the centroids are in $\{-1, 0, 1\}^d$. If we have $2$ (or $-2$) we could change it to $1$ (or $-1$), all vectors which could become farther from the centroid have $2$ in this coordinate. But then the distance for these vectors is still at most $1$. We also may assume that distances are in $\{1, 2, 3\}$, since distance $4$ could be only from $2$ to $-2$.

    We claim that if we set $\delta(v) := \min_{i=1}^2 \distinfty(x_v, c_i)$, $\delta$ is a solution to \probHIOCT. Remove all edges $\{u, v\}$ with $\delta(u) + \delta(v) \ge 2$, and consider 2-coloring of $G$ induced by the partition $\{C_1, C_2\}$. Assume that we have an edge $\{u, v\}$ such that $\delta(u) + \delta(v) \le 1$ and $u$ and $v$ are in the same cluster (w.l.o.g $C_1$). Then we have a coordinate $\{u, v\}$ such that w.l.o.g $x_u[u, v] = 2$ and $x_v[u, v] = -2$, but $\distinfty(x_u, c_1) + \distinfty(x_v, c_1) \le 3$ due to $\delta(u) + \delta(v) \le 1$ and so
    $|x_u[u, v] - c_1[u, v]| + |x_v[u, v]- c_1[u, v]| \le 3$, which is a contradiction. So $(G, t)$ is also a \yesinstance.
\end{proof}

Note that the reduction from \ref{thm:linfoct} also implements \probColoring, if we set $k$ to the number of colors and $D$ to $|V(G)|$, since with such a small budget we can not allow any same-colored neighbors in the optimal clustering.

%!TEX root = Integer_clustering.tex
\section {The case $p \in (1, \infty)$}\label{sec:pinfty}

In this section we consider the case $p \in (1, \infty)$, with the particular emphasis on the most commonly used case $p = 2$. With the $L_2$ distance, the \probClust problem is widely studied under the name \textsc{$k$-Means}.

\subsection{\classFPT when parameterized by $d + D$ for $p = 2$}

When we consider both $d$ and $D$ as the parameters, \probClustSelect in the $L_2$ distance becomes \classFPT, and so \probClust is also \classFPT by Theorem \ref{thm:genfpt}.

Note that in any composite cluster, each vector except at most one is at distance at least $1/4$ from the centroid, so the $\alpha$-property holds with $\alpha=1/4$. Consider two different vectors, they have different values in some coordinate, and in this coordinate at least one of them is at distance at least $(1/2)^2 = 1/4$ from the centroid.

Now we prove Theorem \ref{thm:l2seldD1}, which we restate here.

\ltwoseldD*
\begin{proof}
We start with the proof that  \probClustSelect is \classFPT. Distance $\disttwo$ enjoys the $\alpha$-property. 
   Hence if $t > 4D + 1$ then any composite cluster costs more than $D$ and the instance is clearly a \noinstance. So we may assume that $t \le 4D + 1$.

    We claim that there are at most $4mtD$ possible total weights of the resulting composite cluster. First, in the resulting cluster there could be at most one vector with weight strictly larger than $4D$. Otherwise, let us consider two such vectors and the coordinate in which they differ. No matter which value the centroid has there, it is at distance of at least $1/2$ from at least one of the vectors, so the total cost is larger than $4D(1/2)^2 \ge D$. So there are at most $m$ possibilities for the largest weight, and all of the other $(t - 1)$ weights are at most $4D$.

    We fix the total resulting cluster weight $W$, the vector in the resulting cluster with the largest weight $x_{j^*} \in X_{j^*}$, and the coordinate $i$.
    Since the centroid $c$ is the mean of the vectors in the resulting cluster, $c[i]$ is of form $\frac{y}{W}$, where $y \in \mathbb{Z}$. We claim that the distance from $y$ to $W \cdot x_{j^*}[i]$ is bounded by a function of $D$, and so each possible $y$ could be enumerated in \classFPT time. Moreover, all possible centroids could also be enumerated in \classFPT time since $d$ is a parameter.

    Let $\{x_1, \dots, x_t\}$ be the resulting cluster, $x_j \in X_j$ for all $j \in \{1, \dots, t\}$. The difference between $c[i]$ and $x_{j^*}[i]$ could be written as
\[x_{j^*}[i] - c[i] = x_{j^*}[i] - \sum_{j=1}^t \frac{w(x_j) x_j[i]}{W} = \frac{\sum_{j = 1}^t w(x_j) (x_{j^*}[i] - x_{j}[i])}{W}.\]
    The absolute value of the numerator is $\Oh(D^3)$ since $t = \Oh(D)$, $w(x_{j^*})$ gets multiplied by zero, and all other weights are at most $4D$. Also, for any $j \in \{1, \dots, t\}$, $|x_{j^*}[i] - x_j[i]| \le 4D$, since
\[4D \ge 4 \left((x_{j^*}[i] - c[i])^2 + (x_{j}[i] - c[i])^2\right) \ge (x_{j^*}[i] - x_j[i])^2 \ge |x_{j^*}[i] - x_j[i]|.\]

    The total running time is at most
\[4mtd \cdot m \cdot \Oh(D^3)^d \cdot m,\]
    since we try all possible cluster weights, all possible $x_{j^*}$ out of the input vectors, then all possible centroids which differ from $x_{j^*}$ by $\Oh(D^3)$ in each coordinate. And then for each centroid we check whether the optimal cluster for it has cost at most $D$ by selecting the best $x_j \in X_j$ for each $j \in \{1, \dots, t\}$. This concludes the proof that \probClustSelect  is \classFPT when parameterized by $d+D$.

\medskip

%\end{proof}

%\begin{corollary}
%    \probClust in the $L_2$ distance is \classFPT parameterized by $d + D$.
%    \label{col:l2clustdD}
%\end{corollary}
%\begin{proof}
Now we proceed with the proof that   \probClust  is \classFPT parameterized by $d + D$.
  For that we employ Theorem \ref{thm:genfpt}. We already have the $\alpha$-property and \classFPT algorithm for \probClustSelect.
  % from Theorem \ref{thm:l2seldD1}. 
  Hence the only thing left is to enumerate the set $\mathcal{D}$ of all possible optimal cluster costs not exceeding $D$.

    Since there are $n$ vectors in total, each cluster contains from $1$ to $n$ vectors. For each possible cluster size $s$ the centroid is of the form $\frac{y}{s}$, where $y \in \mathbb{Z}$. Since input vectors have integer coordinates, the cost of any cluster of size $s$ is of form $\frac{z}{s^2}$, where $z \in \mathbb{Z}$. And since the cost is at most $D$, $z \in \{0, \dots, Ds^2\}$. We enumerate all possible cluster sizes in $\{1, \dots, n\}$, and for each cluster size $s$ all possible cluster costs in $\{0/s^2, \dots, Ds^2/s^2\}$. In this way we obtain $\mathcal{D}$, and $|\mathcal{D}| = \Oh(Dn^3)$.
\end{proof}

\subsection{\classW1-hardness when parameterized by $t + D$}

In our setting, \probClust for $p = 2$ seems to be harder than for $p = 1$, since we do not have the nice property that if many vectors have the same value in some coordinate then the centroid must also have this value. On the contrary, even if only one vector diverges from the rest, the optimal centroid also diverges. So the approach with enumerating nontrivial coordinate sets, which we successfully used in the $p \in (0, 1]$ case, is not likely to work.

We are able to prove that \probClustSelect for $p \in (1, \infty)$ is W[1]-hard parameterized by $t + D$. It remains open whether \probClust for $p \in (1, \infty)$ or specifically for $p = 2$ is W[1]-hard or not, but our result shows that at least the approach we used to obtain an \classFPT algorithm in the $p \in (0, 1]$ case would not yield an \classFPT algorithm for $p \in (1, \infty)$.

%\begin{proposition}
%    If we have a cluster, and there are $a$ zeroes, $b$ ones and no other values in some coordinate, then this coordinate contributes $\frac{ab}{a + b}$ to the total cost.
%\end{proposition}
%\begin{proof}
%    The mean for this coordinate is $\frac{b}{a + b}$, and the cost
%    $$a\left(\frac{b}{a + b}\right)^2 + b \left(1 - \frac{b}{a + b}\right)^2 = \frac{ab^2 + ba^2}{(a + b)^2} = \frac{ab}{a + b}.\]
%\end{proof}
%
%\begin{proposition}
%    Assume we have a cluster of $\{0, 1\}$-valued vectors where in one coordinate there are $a_1$ zeroes and $b_1$ ones, and in the other $a_2$ zeroes and $b_2$ ones, and $b_1 \ge b_2$. Then if we consider a cluster where one of the ones is moved from the second coordinate to the first, its cost will be strictly lower.
%    \label{prop:moreones}
%\end{proposition}
%\begin{proof}
%    Denote $c = a_1 + b_1 = a_2 + b_2$. In the new cluster we have $a_1 - 1$ zeroes and $b_1 + 1$ ones in the first coordinate, $a_2 + 1$ zeroes and $b_2 - 1$ ones in the second coordinate. All other coordinates are unchanged, so the cost difference of the old cluster and the new one is
%    $$\frac{a_1b_1}{c} + \frac{a_2b_2}{c} - \frac{(a_1 - 1) (b_1 + 1)}{c} - \frac{(a_2 + 1) (b_2 - 1)}{c} = \frac{(b_1 - b_2) + (a_2 - a_1) + 2}{c} > 0,\]
%    as $b_1 \ge b_2$ and consequently $a_2 \ge a_1$.
%\end{proof}

First we state and prove two technical claims about the geometrical properties of clustering zero-one valued vectors in the $p \in (1, \infty)$ case.

\begin{claim}
    If we have a cluster of size $a + b$ where $a$ vectors have zero and $b$ vectors have one in the coordinate $i$, then the optimal centroid value in this coordinate is equal to
\[\frac{b^\frac{1}{p - 1}}{a^\frac{1}{p - 1} + b^\frac{1}{p - 1}},\]
    and the coordinate $i$ contributes 
\[\frac{ab}{\left(a^\frac{1}{p - 1} + b^\frac{1}{p - 1}\right)^{p - 1}},\]
    to the total cost.
    \label{claim:lpcost}
\end{claim}
\begin{proof}
    Assume that the centroid value in the coordinate $i$ is equal to $c$, then the cost is
\[a c^p + b (1 - c)^p.\]
    It is easy to see that $c < 0$ is worse than $c = 0$, and similarly $c > 1$ is worse than $c = 1$, so we could restrict $c$ to $[0, 1]$. The derivative with respect to $c$ is
\[p (ac^{p - 1} - b (1 - c)^{p - 1}),\]
    as $p > 1$, the derivative is zero if and only if
    \begin{gather*}
        ac^{p - 1} = b(1 - c)^{p - 1}\\
        \left(\frac{c}{1 - c}\right)^{p - 1} = \frac{b}{a}\\
        \frac{c}{1 - c} = \left(\frac{b}{a}\right)^\frac{1}{p - 1}\\
        c = \frac{1}{1 + \left(\frac{a}{b}\right)^\frac{1}{p - 1}} = \frac{b^\frac{1}{p - 1}}{a^\frac{1}{p - 1} + b^\frac{1}{p - 1}}.
    \end{gather*}
    The derivative increases monotonically: when we increase $c$, $c^{p - 1}$ increases and $(1 - c)^{p - 1}$ decreases as $p - 1 > 0$. So the optimal value must be at its unique root defined by the expression above. Thus, the optimal cost is equal to
\[a \frac{b^\frac{p}{p - 1}}{\left(a^\frac{1}{p -1 } + b^\frac{1}{p - 1}\right)^p} + b \frac{a^\frac{p}{p - 1}}{\left(a^\frac{1}{p -1 } + b^\frac{1}{p - 1}\right)^p} = \frac{ab}{\left(a^\frac{1}{p -1 } + b^\frac{1}{p - 1}\right)^{p-1}}.\]
\end{proof}

Now we prove that it is optimal to have as many ones in the same coordinate as possible. For that, we calculate how much each one adds to the total cost depending on how many ones are there in a coordinate.

\begin{claim}
    Consider a cluster of $s$ zero-one valued vectors, denote as $f(b)$ the contribution of a coordinate in which there are $b$ ones and $s - b$ zeroes. The function $f(b) / b$ is strictly decreasing for $0 < b < s$.
    \label{claim:moreones}
\end{claim}
\begin{proof}
        Denote the number of zeroes in the coordinate as $a := s - b$. By Claim \ref{claim:lpcost}, the contribution of the coordinate per each one is
\[\frac{f(b)}{b} = \frac{ab}{\left(a^\frac{1}{p -1 } + b^\frac{1}{p - 1}\right)^{p - 1}} \cdot \frac{1}{b} = \frac{a/s}{\left((a/s)^\frac{1}{p -1 } + (1 - a/s)^\frac{1}{p - 1}\right)^{p - 1}}.\]
        Let us denote $x = a/s$, $0 < x < 1$, the derivative of the above with respect to $x$ is equal to
\[\frac{d}{dx}\left(\frac{x}{\left(x^\frac{1}{p - 1} + (1 - x)^\frac{1}{p - 1}\right)^{p - 1}}\right) = \left(x^\frac{1}{p-1} + (1 - x)^\frac{1}{p - 1}\right)^{-(p-2)} \cdot \left( (1 - x)^\frac{1}{p - 1} + x (1 -x)^{\frac{1}{p - 1} - 1}\right),\]
        which is strictly positive for $0 < x < 1$, hence proving the claim.
\end{proof}

Now we are ready to prove the hardness result, which was stated in the introduction as Theorem \ref{thm:lpselecthard1}. We recall the statement here.

\lpselecthard*
\begin{proof}
    We construct a reduction from \probMultiClique. Given a graph $G$ and a clique size $k$, we construct the following instance of \probClustSelect.

    We set $t$ to $\binom{k}{2}$, each input set of vectors represents a choice of an edge of the clique between two particular colors, so we number them by unordered pairs of indices from 1 to $k$. We set the dimension $d$ to $|V(G)|$, coordinates are numbered by vertices.

    The set $X_{i, j}$ consists of the following vectors: for each edge $\{u, v\} \in E(G)$ between a vertex $u$ of color $i$ and vertex $v$ of color $j$, we add a vector with $1$ in the coordinate $u$ and $1$ in the coordinate $v$, all other coordinates are set to zero. All vectors have weight one. Finally, we set
\[D = k \cdot \frac{(k - 1) \binom{k - 1}{2}}{\left((k - 1)^\frac{1}{p - 1} + \binom{k - 1}{2}^\frac{1}{p - 1}\right)^{p - 1}}.\]
    In Figure \ref{fig:lpselecthard1}, we show the intuition behind the reduction by considering a simple example.

        \begin{figure}[ht]
            \centering
            \begin{subfigure}{0.33\textwidth}
                \centering
                \begin{tikzpicture}[auto,  node distance=2cm,every loop/.style={},thick,main node/.style={circle,draw}, baseline={(0,1)}]
                    \node[main node, color=red] at (0,0) (1) {1};
                    \node[main node, color=blue] (2) [above right of=1] {2};
                    \node[main node,color=blue] (3) [above =0.2cm of 2] {3};
                    \node[main node,color=green] (4) [below right of=2] {4};
        \draw[line width=1pt]
        (1) -- (2);
    \draw[line width=1pt]
    (1) -- (3);
        \draw[line width=1pt]
        (1) -- (4);
        \draw[line width=1pt]
        (2) -- (4);
    \end{tikzpicture}
\end{subfigure}\hfill
            \begin{subfigure}{0.33\textwidth}
                \centering
    \begin{tabular}{c|c}
        &\begin{tabular}{c c c c}
            \textcolor{red}{1}&\textcolor{blue}{2}&\textcolor{blue}{3}&\textcolor{green}{4}
        \end{tabular}\\
        \hline
        $X_{1, 2}$& \begin{tabular}{c c c c}
            \textcolor{red}{1}&\textcolor{blue}{1}&0&0
            \\\textcolor{red}{1}&0&\textcolor{blue}{1}&0
        \end{tabular}\\
        \hline
        $X_{2, 3}$& \begin{tabular}{c c c c}
            0&\textcolor{blue}{1}&0&\textcolor{green}{1}\\
        \end{tabular}\\
        \hline
        $X_{1, 3}$& \begin{tabular}{c c c c}
            \textcolor{red}{1}&0&0&\textcolor{green}{1}\\
        \end{tabular}
    \end{tabular}
\end{subfigure}\hfill
            \begin{subfigure}{0.33\textwidth}
                \centering
    \begin{tabular}{c|c}
        $\{1, 2\}$& \begin{tabular}{c c c c}
            \textcolor{red}{1}&\textcolor{blue}{1}&0&0
        \end{tabular}\\
        \hline
        $\{2, 4\}$& \begin{tabular}{c c c c}
            0&\textcolor{blue}{1}&0&\textcolor{green}{1}\\
        \end{tabular}\\
        \hline
        $\{1, 4\}$& \begin{tabular}{c c c c}
            \textcolor{red}{1}&0&0&\textcolor{green}{1}\\
        \end{tabular}
        \\\midrule
            $c$&\begin{tabular}{c c c c}
                \textcolor{red}{$\frac{2}{3}$}&
                \textcolor{blue}{$\frac{2}{3}$}
                &0&\textcolor{green}{$\frac{2}{3}$}
            \end{tabular}
    \end{tabular}
\end{subfigure}
\caption{An example illustrating the reduction in Theorem \ref{thm:lpselecthard1}: an input graph $G$ colored in three colors, the vector sets produced by the reduction, and the resulting optimal cluster of cost 2, corresponding to the clique on $\{1, 2, 4\}$. Note that in the resulting cluster, each non-zero coordinate has the maximal number of ones, $(k - 1)$.}
            \label{fig:lpselecthard1}
        \end{figure}

            If there is a colorful $k$-clique in $G$ then we construct a solution to our instance of \probClustSelect. Assume the clique is formed by vertices $v_1$, $v_2$, \dots, $v_k$, for each $i \in \{1, \cdots, l\}$ vertex $v_i$ is of color $i$. From each $X_{i, j}$ choose the vector corresponding to the edge $\{v_i, v_j\} \in E(G)$. Among the chosen vectors, in every coordinate of the form $v_i$ there are $(k - 1)$ ones from edges to $v_i$ and $\binom{k}{2} - (k - 1) = \binom{k - 1}{2}$ zeroes. All other coordinates are zeroes in the chosen vectors, so they do not contribute anything to the total distance. By Claim \ref{claim:lpcost}, the total distance is
\[k \cdot \frac{(k - 1) \binom{k - 1}{2}}{\left((k - 1)^\frac{1}{p - 1} + \binom{k - 1}{2}^\frac{1}{p - 1}\right)^{p - 1}} = D.\]

        In the other direction, we prove that only the solution described above could have the cost $D$, all others have strictly larger cost. First notice that in any resulting cluster there are at most $(k - 1)$ ones in each coordinate, since for any vertex $v \in V(G)$, if we denote its color by $i$, only vectors from $(k - 1)$ sets of the form $X_{i, j}$ ($j \in \{1, \dots, k\} \setminus\{i\}$) have ones in the coordinate $v$, and we take one vector from each set by the definition of \probClustSelect.

        Each vector has exactly two ones, so in any resulting cluster there are $2\cdot \binom{k}{2}$ ones in total. By Claim \ref{claim:moreones}, any resulting cluster which does not have $(k - 1)$ ones in $k$ coordinates has strictly larger cost, since only coordinates with exactly $(k - 1)$ ones have the optimal cost per each one.

%    Now we show that any solution which does not have $(k - 1)$ ones in $k$ coordinates has larger cost --- take any such solution, choose two coordinates with the lowest positive number of ones in them (it could not be that one of them has $(k - 1)$ ones, as the total number of ones is divisible by $(k - 1)$) and move one of the ones from the coordinate with the smallest number of ones to the other. By Proposition \ref{prop:moreones} the cost became strictly smaller, and we still have at most $(k - 1)$ ones in each coordinate. Any solution cluster could be transformed to the one where all coordinates have $(k - 1)$ ones by repeating the operation above, each time the cost gets smaller, so having $(k - 1)$ ones in $k$ coordinates is strictly optimal. Note that in this argument it is fine that we do not necessarily get the proper solution to the \probClustSelect on each transformation --- we are only proving that the cost becomes smaller, and it depends only on how many ones are there in each coordinate.

        So, if the resulting cluster has the cost $D$, then there are $k$ coordinates such that in each of them exactly $(k - 1)$ of the chosen vectors have one. We show that in this case the original instance of \probClique has a $k$-clique. For any color $i \in \{1, \dots, k\}$ there are at most $(k - 1)$ ones in all coordinates indexed by vertices of color $i$ in the resulting cluster. So all of these ones are in the same coordinate $v_i$ for some $v_i$. We claim that the vertices $v_1$, \dots, $v_k$ form a clique. Consider vertices $v_i$ and $v_j$, we have taken some vector from $X_{i, j}$, and this vector must have added a one to the coordinates $v_i$ and $v_j$, then by construction the edge $\{v_i, v_j\}$ is in $E(G)$. 

%        Finally, we prove that the ETH bound hold. An $n^{o(t^{1/2} + D^{1/2})}$ algorithm for \probClustSelect together with the reduction above would give a $|V(G)|^{o(k)}$ algorithm for \probClique since $n = |E(G)|$, $t = \binom{k}{2}$ and $D = k \cdot \frac{(k - 1) \binom{k - 1}{2}}{\left((k - 1)^\frac{1}{p - 1} + \binom{k - 1}{2}^\frac{1}{p - 1}\right)^{p - 1}} \le k^2$.
\end{proof}

\fi

%!TEX root = Integer_clustering.tex
\section{Conclusion and open problems}\label{sec:OPEN}

In this paper, we presented an \classFPT algorithm for \probClust with $p \in (0, 1]$ parameterized by $D$. However, for the case $p \in (1, \infty)$ we were able only to show the \classW1-hardness of \probClustSelect. While  intractability of  \probClustSelect does not exclude that \probClust could be \classFPT with $p \in (1, \infty)$, it indicates that the
% approach for proving it  hardness for \probClustSelect shows that to establish parameterized tractability of \probClust  the
  proof of this (if it is true at all) would require an  approach  completely different from ours. 
  %than we used for 
 %case $p \in (0, 1]$. %, but does not in general rule out an \classFPT algorithm for \probClust with $p \in (1, \infty)$. 
  Thus an interesting and  very concrete  open question   concerns the parameterized  complexity of \probClust with $p \in (1, \infty)$ and parameter $D$. 
 
  Another open question is about the fine-grained complexity of \probClust when parameterized by $k + d$. For several distances, we know XP-algorithms: an $\Oh(n^{dk + 1})$ algorithm by Inaba et.~al.~\cite{inaba1994applications} for $p = 2$, as well as  trivial algorithms for $p \in [0, 1]$. For the case when the possible cluster centroids are given in the input, the matching lower bound is shown in~\cite{Cohen-AddadMRR18}. However, we are not aware of a lower bound complementing the algorithmic results in the case 
  when    any point in Euclidean space can serve as a   centroid.

\iffull
Finally, let us note that our \classW1-hardness reductions could  be easily adapted to obtain ETH-hardness results. Our reductions are from \probClique and, assuming ETH, there is no $n^{o(k)}$ algorithm for \probClique. In most of our results, the ETH lower bounds derived from our reductions, can be complemented by   matching upper bounds through a trivial algorithm for \probClustSelect in time $n^{\Oh(d)}$ or $n^{\Oh(t)}$ and, consequently, an algorithm for \probClust obtained by Theorem \ref{thm:genfpt}. However, the reduction in Theorem \ref{thm:lpselecthard1} excludes only a $(nd)^{o(t^{1/2} + D^{1/2})}$ algorithm for \probClustSelect with $p \in (1, \infty)$ under ETH. Both the trivial algorithm in time $n^{\Oh(t)}$ and the algorithm from Theorem \ref{thm:l2seldD1} in time $D^{\Oh(d)}$ (which could also be turned into a $d^{\Oh(D)}$-time algorithm) fail to match this lower bound. So, another open question is, whether there exists a better reduction or a subexponential algorithm could be obtained in this case. 
%, which would be interesting since seemingly none of our and previously known results achieve this kind of running time.
\fi

\bibliography{k-clustering}

\end{document}